\documentclass[11pt]{article}
\usepackage{graphicx,amssymb,amstext,amsmath,multicol,xcolor,color,colortbl,gensymb,arydshln,subcaption,caption,amsthm,bbm,standalone,newclude,bibentry,bm}
\usepackage[margin=1.5cm]{geometry}
\usepackage{natbib}
\usepackage{xr}
\allowdisplaybreaks

\newcommand{\bmrem}{}

\newtheorem{proposition}{Proposition}
\newtheorem{assumption}{Assumption}
\newtheorem{theorem}{Theorem}
\newcommand{\RV}{\mbox{RV}}
\newcommand{\HD}{\textsc{hd}}

\definecolor{paleBlue}{rgb}{0.88,1,1}
\definecolor{grey}{RGB}{220,220,220}

\begin{document}
\title{Determining the Dependence Structure of Multivariate Extremes}
\author{E. S. Simpson$^{1}$, J. L. Wadsworth$^{2}$, J. A. Tawn$^{2}$\vspace{0.25cm}\\ 
$^1$\emph{STOR-i Centre for Doctoral Training, Lancaster University, Lancaster, LA1 4YR, U.K.}\\
$^2$\emph{Department of Mathematics and Statistics, Lancaster University, Lancaster, LA1 4YF, U.K.}}
\date{}

\maketitle

\begin{abstract}
In multivariate extreme value analysis, the nature of the extremal dependence between variables should be considered when selecting appropriate statistical models. Interest often lies with determining which subsets of variables can take their largest values simultaneously, while the others are of smaller order. Our approach to this problem exploits hidden regular variation properties on a collection of non-standard cones and provides a new set of indices that reveal aspects of the extremal dependence structure not available through existing measures of dependence. We derive theoretical properties of these indices, demonstrate their value through a series of examples, and develop methods of inference that also estimate the proportion of extremal mass associated with each cone. We apply the methods to UK river flows, estimating the probabilities of different subsets of sites being large simultaneously.
\end{abstract}



\section{Introduction}\label{sec:intro}
When constructing models in multivariate extreme value analysis, we often need to exploit extremal dependence features. Consider the random vector $X = (X_{1},\dots,X_{d})$, with $X_{i} \sim F_{i}$, as well as a subset of these variables $X_{C} = \left\{X_{i}:i\in C\right\}$, for some $C\in 2^{D}\setminus\emptyset$, i.e.,\ $C$ lies in the power set of $D=\{1,\dots,d\}$ without the empty set. For any $C$ with $|C|\geq 2$, extremal dependence within $X_{C}$ can be summarized by
\begin{align}
	\chi_{C} = \lim_{u\rightarrow 1}\text{pr}\left\{F_{i}(X_{i})>u:i\in C\right\}/(1-u)
\label{eqn:chi}
\end{align}
if the limit exists. In particular, if $\chi_{C} > 0$, the variables in $X_{C}$ are asymptotically dependent, i.e., can take their largest values simultaneously. If $\chi_{C} = 0$, the variables in $X_{C}$ cannot all take their largest values together, although it is possible that for some $\underline{C}\subset C$, $\chi_{\underline{C}}>0$, see for example \cite{Hua2011} or \cite{Wadsworth2013}.

Many models for multivariate extremes are only applicable when data exhibit either full asymptotic dependence, entailing $\chi_{C}>0$ for all $C\in 2^{D}\setminus\emptyset$ with $|C|\geq 2$, or full asymptotic independence, i.e.,\ $\chi_{i,j}=0$ for all $i<j$ \citep{Heffernan2004}. However, often some $\chi_{C}$ are positive whilst others are zero, i.e.,\ only certain subsets of the variables take their largest values simultaneously, while the other variables are of smaller order. The extremal dependence between variables can thus have a complicated structure, which should be exploited when modelling. In this paper, we present two methods for determining this structure.

The full extremal dependence structure is not completely captured by the $2^{d}-d-1$ coefficients $\{\chi_{C}:C\in 2^{D}\setminus\emptyset,|C|\geq 2\}$ since we do not learn fully whether small values of some variables occur with large values of others, or whether individual variables can be extreme in isolation. This is revealed more clearly by decomposing the vector into radial and angular components, $(R,W)$, and examining their asymptotic structure. If the $X_{i}$ follow a common heavy-tailed marginal distribution, usually achieved via a transformation, these pseudo-polar coordinates are defined as $R=\left\|X\right\|_{1}~~ \text{and} ~~W=X\big/\left\|X\right\|_{2}$, for arbitrary norms $\|\cdot\|_{1}$ and $\|\cdot\|_{2}$. We take both to be the $L_{1}$ norm, and assume that $X$ has standard Fr\'{e}chet margins, so that $\text{pr}(X_{i}<x) = \exp\left(-1/x\right)$ for $x>0$ and $i=1,\dots,d$. As such, the radial and angular components are $R = \sum_{i=1}^{d} X_{i} ~~\text{and}~~ W = X/R$, respectively, with $R>0$ and $W\in \mathcal{S}_{d-1}=\big\{(w_{1},\dots,w_{d})\in[0,1]^{d} : \sum_{i=1}^{d}w_{i}=1\big\}$, the $(d-1)$-dimensional unit simplex. It follows that $\text{pr}(R>r)\sim ar^{-1}$ as $r\rightarrow\infty$, for $a\geq 1$, so all the information about extreme events is contained in $W$, and in particular the distribution of $W$ conditioned on $R>r$ as $r\rightarrow\infty$. Under the assumption of multivariate regular variation (\citeauthor{Resnick2010}, \citeyear{Resnick2010}, Chapter~6),
\begin{align}
	\lim_{t\rightarrow\infty}\text{pr}(R>tr,W\in B \mid R>t) = H(B)r^{-1},~~~\text{$r\geq 1,$}
\label{eqn:spectralMeasure}
\end{align}
for $B$ a measurable subset of $\mathcal{S}_{d-1}$, where the limiting spectral measure $H$ satisfies
\begin{align}
	\int_{\mathcal{S}_{d-1}} w_{i} dH(w) = 1/d, ~~~ i=1,\dots,d.
\label{eqn:momentConstraint}
\end{align}
As the radial component becomes large, the position of mass on $\mathcal{S}_{d-1}$ reveals the extremal dependence structure of $X$. We note the link between the dependence measure $\chi_{C}$ in~\eqref{eqn:chi}, and the spectral measure $H$: if $\chi_{C}>0$, then $H$ places mass on at least one region $\mathcal{S}_{d-1}^{\overline{C}} = \left\{(w_{1},\dots,w_{d})\in[0,1]^{d} : \sum_{i\in \overline{C}}w_{i}=1\right\}$, with $C\subseteq \overline{C}\subseteq D$. This underlines that the term asymptotic dependence is not so useful here, since it offers only partial insight into the structure. In what follows, we thus avoid this term where possible, talking instead about faces of the simplex on which $H$ places mass.

\begin{figure}
\begin{center}
\includegraphics[width=0.5\textwidth]{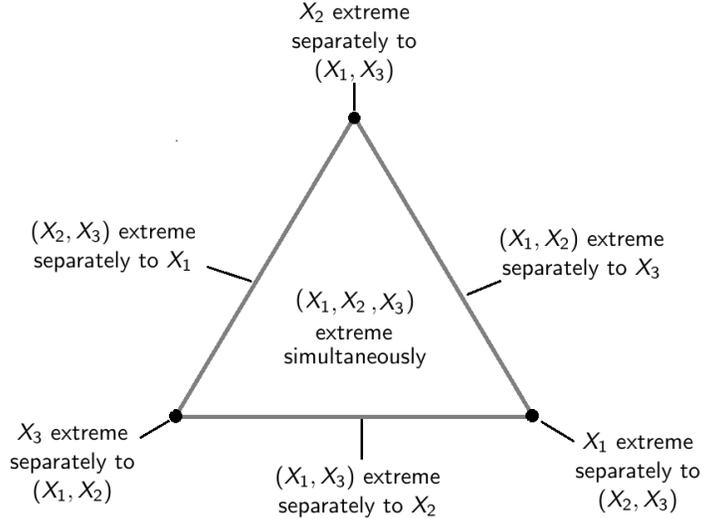}
\caption{The simplex $\mathcal{S}_{2}$. Coordinates are transformed to the equilateral simplex.}
\label{fig:triangleSimplex}
\end{center}
\end{figure}

In the $d$-dimensional case, $\mathcal{S}_{d-1}$ can be partitioned into $2^{d}-1$ faces, each of which could contain mass. Mass on each of these faces corresponds to a different subset of the variables $(X_{1},\dots,X_{d})$ being the only ones taking their largest values concurrently. This is demonstrated in Fig.~\ref{fig:triangleSimplex} for $d=3$. For high, or even moderate, dimensions, there are many faces to consider, and the task of determining which faces truly contain mass, and therefore the extremal dependence structure of the variables, is not straightforward, as for a finite sample with continuous margins, points cannot lie exactly on the boundary of the simplex: no $W_{i}$ equal zero when $R<\infty$. 

The multivariate regular variation assumption~\eqref{eqn:spectralMeasure} can also be phrased in terms of measures on the cone $\mathbb{E}=[0,\infty]^d\setminus\{0\}$, see Section~\ref{subsec:MRV}. Each face of $\mathcal{S}_{d-1}$ can be identified with a sub-cone of $\mathbb{E}$ for which one or more components are identically zero. Intuition and visualization are often simpler with $H$, but in the sequel we work with $\mathbb{E}$ and the sub-cones corresponding to faces of $\mathcal{S}_{d-1}$. Variants of our methods that directly use the radial-angular framework are presented in the 2019 Lancaster University PhD thesis of E.\ S.\ Simpson.

The problem of determining the extremal dependence structure of variables has been recently studied elsewhere in the literature. Under the assumption that the data are from an elliptical copula, \cite{Kluppelberg2015} use factor analysis on extreme correlations linked to the tail dependence function. \cite{Chautru2015} introduces a non-parametric approach based on statistical learning, combining a principal component analysis algorithm with clustering techniques. A Bayesian clustering method is proposed by \cite{Vettori2017}, based on the hierarchical dependence structure of the nested logistic distribution of \cite{Tawn1990}. \citeauthor{Goix2016}\ (\citeyear{Goix2016}, \citeyear{Goix2015}) propose a non-parametric simplex partitioning approach in which they condition on the radial variable being above some high threshold. They assume that there is mass on a particular face if the number of points in the corresponding region of the simplex is sufficiently large, leading to a sparse representation of the dependence structure. \cite{Chiapino2017} propose an algorithm to group together nearby faces with extremal mass into feature clusters, by exploiting their graphical structure and a measure of extremal dependence. Finally, \cite{Chiapino2018} extend this approach by instead using the coefficient of tail dependence of \cite{Ledford1996}.

In this paper we exploit additional, but commonly satisfied, hidden regular variation assumptions on non-standard sub-cones of $\mathbb{E}$, by introducing a new set of parameters that describes the dominant extremal dependence structure. We study properties of these parameters, their link to existing coefficients, and explore their values for a range of examples. Estimation of the parameters provides us with an asymptotically-motivated framework for determining the extremal dependence structure, as well as allowing us to estimate the proportion of mass associated with each set of variables. We propose two such inferential methods, both with computational complexity $O(dn\log n)$ if $d<n$, for $d$ representing the number of variables and $n$ the number of data points. This is the same complexity as the method of \cite{Goix2015}.


\section{Theoretical motivation}\label{sec:theory}
\subsection{Multivariate regular variation}\label{subsec:MRV}
A function $\lambda : (0,\infty] \rightarrow (0,\infty]$ is said to be regularly varying at infinity, with index $\alpha\in\mathbb{R}$, if $\lambda(tx)/\lambda(t) \rightarrow x^{\alpha}$, as $t\rightarrow\infty$, for all $x>0$. For such functions, we write $\lambda \in {\rm{RV}}_{\alpha}$. We can always express $\lambda(x) = L(x)x^{\alpha}$, with $L\in {\rm{RV}}_{0}$ termed a slowly varying function. A cone $G\subset\mathbb{R}^{d}$ is a set such that for any $x\in G$, $tx\in G$ for all $t>0$. The assumption of multivariate regular variation on the cone $G$ means that there exists a scaling function $a(t)\rightarrow\infty$, and a positive measure $\mu$, such that 
\begin{align}
	t~\text{pr}\left(X/a(t)\in \cdot\right) \rightarrow \mu(\cdot),~~~~t\rightarrow\infty,
	\label{eqn:MRV}
\end{align}
with vague convergence in the space of non-negative Radon measures on $G$ (\citeauthor{Resnick2010} \citeyear{Resnick2010}, Chapter 3). If we assume that the margins of $X$ are standard Fr\'{e}chet or Pareto, we may take $a(t)=t$, and the limit measure $\mu$ is homogeneous of order $-1$. For the remainder of this section, we assume that $X$ has standard Fr\'{e}chet marginal distributions.

\subsection{Hidden regular variation}\label{subsec:HRV}
The concept of hidden regular variation was introduced by \cite{Resnick2002}, who formalized and extended the ideas of \citeauthor{Ledford1997} (\citeyear{Ledford1996}, \citeyear{Ledford1997}). Further work has been done by \cite{Maulik2004} and \cite{Mitra2011}, for example, whilst \cite{Resnick2010} provides a textbook treatment. Here, multivariate regular variation is assumed on some cone in $\mathbb{R}^{d}$. If there is also regular variation, but with a scaling function of smaller order, on some sub-cone, we have hidden regular variation on that sub-cone.

To our knowledge, the marginal case of this hidden regular variation framework is the only one previously exploited from a statistical perspective; from a theoretical viewpoint, \cite{Das2013} consider hidden regular variation on a series of non-standard cones, although these are mostly different from the ones we will consider. For $X_{C}=\{X_{i}:i\in C\}$,  $x_C=\{x_i:i\in C\}$, \cite{Ledford1997} considered multivariate regular variation on the cone $\mathbb{E}=[0,\infty]^d\setminus \{0\}$ and hidden regular variation on
\begin{align}
	\mathbb{E}_{C}^* = \left\{ x\in \mathbb{E} : x_{C}\in (0,\infty]^{|C|},~ x_{D\setminus C} \in [0,\infty]^{|D\setminus C|} \right\},
\label{eqn:estar}
\end{align}
with limit measures on $\mathbb{E}^*_C$ homogeneous of order $-1/\eta_C$, and the so-called coefficient of tail dependence $\eta_C$ taking values in $(0,1]$. If $\mu(\mathbb{E}_{C}^*)>0$, variables $X_C$ can take their largest values simultaneously. If we instead consider sub-cones of $\mathbb{E}$ of the form
\begin{align}
	\mathbb{E}_C = \left\{ x\in \mathbb{E} : x_C\in (0,\infty]^{|C|},~ x_{D\setminus C} = \{0\}^{|D\setminus C|}\right\},
	\label{eqn:subcones}
\end{align}
where $\{0\}^m$ denotes an $m$-vector of zeros, then having $\mu(\mathbb{E}_C)>0$ indicates that variables in $X_{C}$ can take their largest values simultaneously while variables in $X_{D\setminus C}$ are of smaller order. Our task is to determine the cones $\mathbb{E}_C$ on which $\mu$ places mass, equivalent to the problem of detecting where $H$ places mass on $\mathcal{S}_{d-1}$, thus revealing the extremal dependence structure of $X$. For simplicity, we assume that if $\mu$ places mass on $\mathbb{E}_C$, then for any measurable $B_C\subset\mathbb{E}_C$, $\mu(B_C)>0$. More generally, it need only be true that there exists $B_C\subset\mathbb{E}_C$ such that $\mu(B_C)>0$, however if the mass lies only in very restricted parts of $\mathbb{E}_C$, then the task of detecting which cones contain mass is naturally more difficult.

For a finite sample, mass will not occur on cones $\mathbb{E}_C$ with $|C| < d$; this is equivalent to all mass for $W$ being placed on the interior of the simplex in Fig.~1. One option is to truncate the variables to zero below some marginal threshold, to ensure mass on at least some of these cones at a finite level. Let us define
\begin{align}
	X^{*} = \begin{cases} 0, & \mbox{$X\leq -1/\log p$}, \\ X, & \mbox{$X> -1/\log p$}, \end{cases}
\label{eqn:truncation}
\end{align}
such that $p$ is the quantile at which we truncate. The variable $X^*$ has the same tail behaviour as $X$, but in general $\text{pr}(X^*/t \in B_C)>0$ for $B_C \subset \mathbb{E}_C$, and in this way we could define a hidden regular variation assumption on $\mathbb{E}_C$. Writing $B_C = \{x \in \mathbb{E}: x_C \in B \subset (0,\infty]^{|C|}, x_{D \setminus C}\in\{0\}^{|D \setminus C|}\}$, then $\text{pr}(X^*/t \in B_C) = \text{pr}(X_C/t \in B, X_{D \setminus C} \in [0, -1/\log p]^{|D \setminus C|})$, such that we consider the behaviour when the variables $X_C$ are growing at a common rate, but variables $X_{D \setminus C}$ have a fixed upper bound. However, the latter condition does not capture all possible behaviour that leads to variables $X_{D \setminus C}$ being of smaller order than $X_C$, and in general a more elaborate assumption is needed. We consider how we can allow $X_{D\setminus C}$ to be bounded above by a function that is growing, but at a potentially slower rate than $t$.

Define the set $(y, \infty]^C \times [0,z]^{D \setminus C} = \{x \in \mathbb{E}: x_i > y, i \in C; x_j \leq z, j \in D \setminus C\}$. Then under the regular variation assumption~\eqref{eqn:MRV}, \begin{align}
t\text{pr}\left\{X/t \in (y,\infty]^{C} \times [0,z]^{D \setminus C}\right\} \to \mu\left((y,\infty]^C \times [0,z]^{D \setminus C}\right) \geq \mu\left((y,\infty]^C \times \{0\}^{D \setminus C}\right). 
\label{eqn:muconv1}
\end{align} 
Therefore, if $\mu\left((y,\infty]^C \times \{0\}^{D \setminus C}\right)>0$, and hence $\mu\left(\mathbb{E}_C\right)>0$, this indicates that in~\eqref{eqn:muconv1} we may be able to consider $z=z_t \to 0$ at a suitable rate in $t$ and still observe 
\begin{align}
 \lim_{t \to \infty} t\text{pr}\left\{X/t \in (y,\infty]^{C} \times [0,z_t]^{D \setminus C}\right\} >0. \label{eqn:mu-1}
\end{align}
A consequence of a positive limit in~\eqref{eqn:mu-1} is that $\text{pr}\{X/t \in (y,\infty]^{C} \times [0,z_t]^{D \setminus C}\} \in \RV_{-1}$. When the limit in~\eqref{eqn:mu-1} is zero, then either $z_t \to 0$ too quickly --- consider for example the case $z_t \equiv 0$ --- or $\mu$ places no mass on $\mathbb{E}_C$. In these cases we focus on the rate of convergence to zero in~\eqref{eqn:mu-1}. Taking $z_t = z t^{\delta-1}$ for $\delta \in [0,1]$, and rephrasing in terms of min and max projections, our main assumption is as follows.
\begin{assumption}\label{assumption:HRV}
Suppose we have regular variation on the cone $\mathbb{E}=[0,\infty]^d\setminus \{0\}$, so that equation~\eqref{eqn:MRV} is satisfied with $\mu$ homogeneous of order $-1$. For all $C\subseteq D$, let $X_\wedge^C = \min_{i\in C} X_i$ and $X_\vee^{D\setminus C} = \max_{i\in D\setminus C}X_i$. We assume that for all $\delta\in[0,1]$,
\begin{align}
\text{pr}\left\{\left(X_\wedge^C/t,X_\vee^{D\setminus C}/t^\delta\right)\in (y,\infty]\times[0,z]\right\} \in {\rm{RV}}_{-1/\tau_C(\delta)}, ~~~\text{$t\rightarrow\infty$},~~~\text{$0<y,z<\infty$},\label{eqn:RVassumption1}
\end{align}
and that there exists $\delta^*<1$ such that $\tau_C(\delta^*)=1$ for all $C$ such that $\mu(\mathbb{E}_C)>0$, and $\tau_C(\delta^*)<1$ for all $C$ such that $\mu(\mathbb{E}_C)=0$.
\end{assumption}
\noindent We note that the probability in~\eqref{eqn:RVassumption1}, and hence $\tau_C(\delta)$, is non-decreasing in $\delta$. The case $\delta=0$ and $z=-1/\log p$ is identical to a regular variation assumption on the truncated variables $X^*$; allowing $\delta>0$ produces a more diverse range of possibilities. 

Through the final line of Assumption~\ref{assumption:HRV}, the indices $\tau_C(\delta)$ contain information on the limiting extremal dependence structure; the challenge is to find a suitable $\delta^*$, noting that if $\delta$ is too small we could have $\tau_C(\delta)<1$ even when $\mu(\mathbb{E}_C)>0$, but if $\delta$ is too large, some $\tau_C(\delta)$ could be close to one even when $\mu(\mathbb{E}_C)=0$, making the detection problem difficult in light of statistical uncertainty. These issues are discussed further below and in Section~\ref{sec:methodology}.

 Overall, examining the regular variation properties in Assumption~\ref{assumption:HRV} leads to understanding of the sub-asymptotic behaviour of $\mu$ in relation to which cones $\mathbb{E}_C$ are charged with mass. This is analogous to determining the support of $H$ in~\eqref{eqn:spectralMeasure}. In the remainder of this section, we illustrate the utility and validity of our hidden regular variation assumption via examples, and discuss properties of $\tau_{C}(\delta)$. Theorems~\ref{thm:etatau1}~and~\ref{thm:etatau2} clarify some links between $\tau_C(\delta)$ and $\eta_C$.

\begin{theorem}
Assume regular variation, or hidden regular variation, on $\mathbb{E}_C^*$ defined in \eqref{eqn:estar}, such that $\text{pr}\left(X^{C}_\wedge>t\right)\in{\rm{RV}_{-1/\eta_C}}$. Suppose further that for all $\overline{C}\subseteq D$ such that $\overline{C} \supseteq C$, Assumption~\ref{assumption:HRV} is satisfied, so that for $\delta=1$, $\text{pr}\left(X^{\overline{C}}_\wedge>t,X_\vee^{D\setminus \overline{C}}\leq t\right)\in{\rm{RV}}_{-1/\tau_{\overline{C}}(1)}$. Then $\eta_C = \max_{\overline{C}: C \subseteq \overline{C}} \tau_{\overline{C}}(1)$.
\label{thm:etatau1}
\end{theorem}
\begin{proof}
We have
\[\text{pr}\left(X^{C}_\wedge>t\right) = \text{pr}\left(X/t\in(1,\infty]^{C}\times[0,\infty]^{D\setminus C}\right);\]
\[\text{pr}\left(X^{\overline{C}}_\wedge>t,X_\vee^{D\setminus\overline{C}}\leq t\right) = \text{pr}\left(X/t\in(1,\infty]^{\overline{C}}\times[0,1]^{D\setminus\overline{C}}\right).
\] 
By the partition $(1,\infty]^C\times[0,\infty]^{D\setminus C}=\bigcup_{\overline{C}:C\subseteq \overline{C}}(1,\infty]^{\overline{C}}\times[0,1]^{D\setminus\overline{C}}$, we deduce that
\[
\text{pr}\left(X/t\in(1,\infty]^{C}\times[0,\infty]^{D\setminus C}\right) = \sum\limits_{\overline{C}:C\subseteq\overline{C}}\text{pr}\left(X/t\in(1,\infty]^{\overline{C}}\times[0,1]^{D\setminus\overline{C}}\right),
\]
from which the result follows.
\end{proof}
We note $\tau_D(\delta)$ does not depend on $\delta$; we therefore denote it by $\tau_D$.
\begin{theorem}
For all $C\in 2^{D}\setminus\emptyset$ with $|C|\geq 2$, assume regular variation or hidden regular variation on $\mathbb{E}_{C}^*$ with coefficient of tail dependence $\eta_{C}$, and suppose Assumption~\ref{assumption:HRV} holds. For any set $C$ with $|C|\geq 2$, assume that for any $\overline{C}\supseteq C$, $\eta_{\overline{C}'} < \eta_{\overline{C}}$ for all $\overline{C}'\supset \overline{C}$. Then $\tau_{C}(1)=\eta_{C}$, and $\tau_C(\delta)\leq\eta_C$ for all $\delta\in[0,1]$.
\label{thm:etatau2}
\end{theorem}
\begin{proof}
Since $\mathbb{E}_D=\mathbb{E}_D^*$, we have $\eta_{D}=\tau_{D}$, and by Theorem~\ref{thm:etatau1}, for any set $C_{d-1}\subset D$ with $|C_{d-1}|=d-1$,
\[
	\eta_{C_{d-1}} = \max\left\{\tau_{C_{d-1}}(1),\tau_{D}\right\}= \max\left\{\tau_{C_{d-1}}(1),\eta_{D}\right\}.
\]
Since, by assumption, $\eta_{{C}_{d-1}}>\eta_{D}$, we have $\eta_{C_{d-1}}=\tau_{C_{d-1}}(1)$. Similarly, for any set $C_{d-2}\subset C_{d-1}$ with $|C_{d-2}|=d-2$,
\[
\eta_{C_{d-2}}=\max\left\{\tau_{{C}_{d-2}}(1),\tau_{C_{d-1}}(1),\tau_{D}\right\} =\max\left\{\tau_{{C}_{d-2}}(1),\eta_{C_{d-1}},\eta_{D}\right\}.
\]
Again, since $\eta_{C_{d-2}}>\eta_{C_{d-1}}>\eta_{D}$ for all $C_{d-2}\subset C_{d-1}$, then $\eta_{C_{d-2}} =\tau_{{C}_{d-2}}(1)$. The result $\tau_C(1)=\eta_C$ follows by iteration for any set $C$ with $|C|\geq 2$. Since $\tau_C(\delta)$ is non-decreasing in $\delta$, $\tau_C(1)=\max_{\delta\in[0,1]} \tau_C(\delta)$, so $\tau_C(\delta)\leq\eta_C$ for all $\delta\in[0,1]$.
\end{proof}

In the Appendix and Supplementary Material, respectively, we calculate the value of $\tau_{C}(\delta)$, with $C\in 2^{D}\setminus \emptyset$, for a range of bivariate and multivariate copulas. For the bivariate case, we restrict our investigation to a subclass of bivariate extreme value distributions \citep{Tawn1988} that covers all possible combinations of cones $\mathbb{E}_C$ charged with mass, focusing on the case where the spectral density is regularly varying at 0 and 1. For multivariate cases there are many more possibilities, so we study certain trivariate extreme value distributions \citep{Tawn1990}, which have $\chi_C>0$ for at least one set $|C|\geq 2$, and two classes of copula having $\chi_C=0$ for all $|C|\geq 2$. The results are summarized here.

The bivariate extreme value distribution in standard Fr\'{e}chet margins has distribution function of the form $F(x,y) = \exp\left\{-V(x,y)\right\}$
for some exponent measure
\begin{align}
        V(x,y) = 2\int_{0}^{1} \max\left\{w/x,(1-w)/y\right\}dH(w),~~~\text{$x,y>0$},
\label{eqn:bveDist}
\end{align}
where $H$ denotes the spectral measure defined in equation~\eqref{eqn:spectralMeasure} on the unit simplex $[0,1]$. In the bivariate case, $\mathbb{E}=[0,\infty]^2\setminus\{0\}$ can be partitioned into three natural cones: $\mathbb{E}_1$, $\mathbb{E}_2$ and $\mathbb{E}_{1,2}$. If $H(\{0\})=\theta_{2}\in[0,1/2]$ and $H(\{1\})=\theta_{1}\in[0,1/2]$, the distribution places mass $\theta_2$, $\theta_1$, $\theta_{1,2}=1-(\theta_1+\theta_2)$ in the three cones. If $\theta_1+\theta_2=1$, the variables are independent, and $\mu(\mathbb{E}_{1,2})=0$. In this case, all the limiting mass is placed on  $\mathbb{E}_1$ and $\mathbb{E}_2$. Here, Assumption~\ref{assumption:HRV} holds for $C=\{1\},\{2\}$ and $\{1,2\}$ with $\tau_{1}(\delta)=\tau_{2}(\delta)=1$ for all $\delta\in[0,1]$, and $\tau_{1,2}=\eta_{1,2}=1/2$.

When $\theta_1+\theta_2<1$, $\mu(\mathbb{E}_{1,2})>0$ and $\tau_{1,2}=\eta_{1,2}=1$, i.e., both variables can be simultaneously large. If $\theta_i>0$, it follows that $\tau_i(\delta)=1$ for $\delta\in[0,1]$ and $i=1,2$, and there is mass on the corresponding cone $\mathbb{E}_i$. However, when $\theta_1=\theta_2=0$, there is no mass on either of these cones, and additional conditions are required for~\eqref{eqn:RVassumption1} to hold. We suppose that $H$ is absolutely continuous on $(0,1)$ with Lebesgue density $h(w)=dH(w)/dw$ satisfying $h(w)\sim c_{1}(1-w)^{s_{1}}$ as $w\rightarrow 1$, and $h(w)\sim c_{2}w^{s_{2}}$ as $w\rightarrow 0$, for $s_1, s_2>-1$ and $c_1, c_2 >0$. In the Appendix, we show that for $i=1,2$, $\tau_{i}(\delta)=\left\{(s_i+2)-\delta(s_i+1)\right\}^{-1}$. To illustrate this final case, consider the bivariate extreme value distribution with the logistic dependence structure \citep{Tawn1988}, with $V(x,y)=\left(x^{-1/\alpha}+y^{-1/\alpha}\right)^{\alpha}$ and
\begin{align}
        h(w) = \frac{1}{2}\left(\alpha^{-1}-1\right)\left\{w^{-1/\alpha} + (1-w)^{-1/\alpha}\right\}^{\alpha-2}\left\{w(1-w)\right\}^{-1-1/\alpha},
\label{eqn:logisticDensity}
\end{align}
$0<w<1$, $\alpha\in(0,1)$. For this model $s_{1}=s_{2}=-2+1/\alpha$, and so $\tau_{1}(\delta)=\tau_{2}(\delta)=\alpha/(1+\alpha\delta-\delta)$ which increases from $\tau_i(\delta)=\alpha<1$ at $\delta=0$ to $\tau_i(\delta)=1$ at $\delta=1$.

\begin{table}[ht]
\resizebox{\textwidth}{!}{%
\centering
\begin{tabular}{|c|c|c|c|}
\hline
& $|C|=1$ & $|C|=2$ & $|C|=3$\\
\hline
(i) & $\tau_{1}(\delta)=\tau_{2}(\delta)=\tau_{3}(\delta)=1$ & $\tau_{1,2}(\delta)=\tau_{1,3}(\delta)=\tau_{2,3}(\delta)=1/2$ &$\tau_{1,2,3}=1/3$\\
(ii) & $\tau_{1}(\delta)=\tau_{2}(\delta)=\frac{\alpha}{1+\alpha\delta-\delta}, \tau_{3}(\delta)=1$ & $\tau_{1,2}(\delta)=1, \tau_{1,3}(\delta)=\tau_{2,3}(\delta)=\frac{\alpha}{\alpha\delta+1+\alpha-\delta}$ & $\tau_{1,2,3}=1/2$\\
(iii) & $\tau_{1}(\delta)=\tau_{2}(\delta)=\tau_{3}(\delta)=\frac{\alpha}{1+\alpha\delta-\delta}$ & $\tau_{1,2}(\delta)=\tau_{1,3}(\delta)=\tau_{2,3}(\delta)=\frac{\alpha}{2+\alpha\delta-2\delta}$ &$\tau_{1,2,3}=1$\\
(iv) & $\tau_{1}(\delta)=\tau_{2}(\delta)=\tau_{3}(\delta)=1$ & $\tau_{1,2}(\delta)=\tau_{1,3}(\delta)=\tau_{2,3}(\delta)=2^{-\alpha}$ & $\tau_{1,2,3}=3^{-\alpha}$\\
\hline
\end{tabular}}
\caption{Values of $\tau_{C}(\delta)$ for some trivariate copula examples. For all logistic models, the dependence parameter $\alpha$ satisfies $0< \alpha <1$, with larger $\alpha$ values corresponding to weaker dependence. Case~(i): independence; case~(ii): independence and bivariate logistic; case~(iii): trivariate logistic; case~(iv): trivariate inverted logistic.}
   \label{table:tauExamples}
\end{table}

When $d=3$, there are many more possibilities for combinations of cones $\mathbb{E}_C$ with mass. Table~\ref{table:tauExamples} gives $\tau_C(\delta)$ for four examples, in each case identifying $\tau_C(\delta)$ on cones $\mathbb{E}_C$ with $|C|=1,2,3$. Cases (i)-(iii) in Table~\ref{table:tauExamples} are all special cases of the trivariate extreme value copula. Case (i) is the independence copula, which has limit mass on $\mathbb{E}_1$, $\mathbb{E}_2$ and $\mathbb{E}_3$. For the $d$-dimensional independence copula, $\tau_C(\delta)=|C|^{-1}$ for $|C|\leq d$, and does not depend on the value of $\delta$. Case (ii) is the copula corresponding to variables $(X_1,X_2)$ following a bivariate extreme value logistic distribution~\eqref{eqn:logisticDensity}, independent of $X_3$. Here all the limit mass is placed on $\mathbb{E}_{1,2}$ and $\mathbb{E}_3$. Again, $\tau_C(\delta)$ differs between cones where $C$ is of different dimension. The trivariate extreme value logistic model, case (iii), places all extremal mass on $\mathbb{E}_{1,2,3}$, so that $\tau_{1,2,3}=1$, and $\tau_C(\delta)<1$ when $\delta<1$, for all $C$ with $|C|<3$. Since this is a symmetric model, $\tau_C(\delta)$ is the same on all three cones $\mathbb{E}_C$ with $|C|=1$, and is also equal for each cone with $|C|=2$. 

Copula (iv) is the inverted extreme value copula (\citeauthor{Ledford1997}, 1997), with a symmetric logistic dependence model. It places all limiting mass on cones $\mathbb{E}_C$ with $|C|=1$, but unlike the independence copula, has sub-asymptotic dependence, reflected by the values of $\tau_C(\delta)$ for cones $\mathbb{E}_C$ with $|C|=2,3$, which are closer to one than in the independence case. The values of $\tau_C(\delta)$ do not depend on $\delta$ for this case.

The Gaussian copula with covariance matrix $\Sigma$ also exhibits asymptotic independence, with all limit mass on $\mathbb{E}_C$ with $|C|=1$. We study the values of $\tau_C(\delta)$ for the trivariate case in the Supplementary Material. For sets $C=\{i\}$, $i=1,2,3$, $\tau_C(\delta)=1$ only if $\delta\geq \max\left(\rho^2_{ij},\rho^2_{ik}\right)$, where $\rho_{ij}>0$ is the Gaussian correlation parameter for variables $i$ and $j$; otherwise, $\tau_C(\delta)<1$. We also know that $\tau_{1,2,3}=\eta_{1,2,3}=\left(1_3^T\Sigma^{-1}1_3\right)^{-1}$, with $1_d\in\mathbb{R}^d$ denoting a vector of 1s. For $C=\{i,j\}$, $i< j$, under Assumption \ref{assumption:HRV}, Theorem~2 leads to $\tau_C(1)=\eta_C=\left(1_2^T\Sigma_{i,j}^{-1}1_2\right)^{-1}$, with $\Sigma_{i,j}$ denoting the submatrix of $\Sigma$ corresponding to variables $i$ and $j$, provided the correlations satisfy $1+\rho_C\neq \sum\limits_{C':|C'|=2,C'\neq C}\rho_{C'}$; for $\delta<1$, $\tau_C(\delta)\leq\eta_C$. 


\section{Methodology}\label{sec:methodology}
\subsection{Introduction to methodology}
The coefficient $\tau_C(\delta)$ defined in Assumption~\ref{assumption:HRV} reveals whether the measure $\mu$ places mass on the cone $\mathbb{E}_C$. For $\mu(\mathbb{E}_C)>0$, we assume there exists $\delta^*<1$ such that $\tau_C(\delta^*)=1$, but we could still have $\tau_C(\delta)<1$ for values of $\delta<\delta^*$. For cones with $\mu(\mathbb{E}_C)=0$, the detection problem becomes easier the further $\tau_C(\delta)$ is from 1, and since $\tau_C(\delta)$ is non-decreasing in $\delta$ it is ideal to take $\delta$ as small as possible. We therefore have a trade-off between choosing $\delta$ large enough that $\tau_C(\delta)=1$ on cones $\mathbb{E}_C$ with extremal mass, but small enough that $\tau_C(\delta)$ is not close to 1 on those $\mathbb{E}_C$ without extremal mass. For the examples in Table~\ref{table:tauExamples}, we could take $\delta=0$, since the cones with $\mu(\mathbb{E}_C)>0$ have $\tau_C(\delta)=1$ for all $\delta\in[0,1]$. However, the Gaussian case reveals that although $\mu(\mathbb{E}_i)>0$, for $i=1,\dots,d$, we can have $\tau_i(0)<1$, so it is necessary to take $\delta>0$ for correct identification of cones with extremal mass.

We therefore introduce two approaches for determining the extremal dependence structure of a set of variables. In the first method we set $\delta=0$, and apply a truncation to the variables $X$ by setting any values below some marginal threshold equal to zero. This transformation is analogous to the approach of \cite{Goix2015}, who partition the non-negative orthant in a similar way, but we additionally exploit Assumption~\ref{assumption:HRV}. In our second method, we consider $\delta>0$ when exploiting the regular variation assumption. As well as aiming to determine the extremal dependence structure, both methods estimate the proportion of extremal mass associated with each cone $\mathbb{E}_C$.

\subsection{Method 1: $\delta=0$}\label{subsec:method1}
We apply Assumption~\ref{assumption:HRV} with $\delta$ equal to zero by applying truncation~\eqref{eqn:truncation} to variables $X$ for some choice of $p$. Recall that the cone $\mathbb{E}$ equals $\bigcup_{C\in 2^{D}\setminus\emptyset}\mathbb{E}_{C}$, with the components of the union disjoint and defined as in \eqref{eqn:subcones}. We wish to partition $\mathbb{E}$ with approximations to $\mathbb{E}_{C}$, by creating regions where components indexed by $C$ are large and those not in $C$ are small. This is achieved via regions of the form
\begin{align}
	E_C = \left\{x^*\in\mathbb{E}: x_C^* \in(-1/\log p,\infty]^{|C|},~x_{D\setminus C}^* \in\{0\}^{|D\setminus C|}\right\}.
\label{eqn:regions}
\end{align}
Define the variable $Q=\min\left(X_i^*:X_i^*>0,i=1,\dots,d\right)$, and recall that we denote $X_\wedge^C = \min_{i\in C} X_i$ and $X_\vee^{D\setminus C} = \max_{i\in D\setminus C}X_i$. Under Assumption~\ref{assumption:HRV}, as $q\rightarrow\infty$,
\[
	\text{pr}\left(Q > q\mid X^*\in E_C\right)\propto \text{pr}\left(X_\wedge^C>q, X_\vee^{D\setminus C}<-1/\log p\right) \in{\rm{RV}}_{-1/\tau_C(0)},
\]
so that $\text{pr}\left(Q > q\mid X^*\in E_C\right)=L_C(q)q^{-1/\tau_C(0)}$ for some slowly varying function $L_C$. We now let $\tau_C=\tau_C(0)$, and assume that the model
\begin{align}
	\text{pr}\left(Q>q\mid X^*\in E_C\right) = K_Cq^{-1/\tau_C}, ~~~~q>u_C,
\label{eqn:HRVassumption}
\end{align}
holds for a high threshold $u_C$, with $\tau_C\in(0,1]$ and $K_C>0$ for all $C\in 2^D\setminus\emptyset$. Here, the slowly varying function $L_{C}$ is replaced by the constant $K_{C}$ as a modelling assumption, removing the possibility of having $L_{C}(q)\rightarrow 0$ as $q\rightarrow\infty$. 

Model~\eqref{eqn:HRVassumption} may be fitted using a censored likelihood approach. Suppose that we observe $n_C$ values $q_{1},\dots q_{n_{C}}$ of $Q$ in $E_{C}$. The censored likelihood associated with $E_C$, is
\begin{align}
L_{C}(K_{C},\tau_{C}) = \prod_{j=1}^{n_{C}} \Big(1 - K_{C}u_{C}^{-1/\tau_{C}} \Big)^{\mathbbm{1}_{\{q_{j}\leq u_{C}\}}} \bigg( \frac{K_{C}}{\tau_{C}}q_{j}^{-1-1/\tau_{C}}\bigg)^{\mathbbm{1}_{\{q_{j}> u_{C}\}}},
\label{eqn:censored}
\end{align}
with $u_{C}$ a high threshold. Analytical maximization of \eqref{eqn:censored} leads to closed-form estimates of $(K_{C},\tau_{C})$, with the latter corresponding to the Hill estimate \citep{Hill1975}. In particular,
\begin{align*}
	\hat{\tau}_C = \left(\sum\limits_{j=1}^{n_C}\mathbbm{1}_{\{q_j>u_C\}}\right)^{-1}\sum\limits_{j=1}^{n_C}\mathbbm{1}_{\{q_j>u_C\}}\log\left(\frac{q_j}{u_C}\right),~~~ \hat{K}_C = \left(\frac{\sum_{j=1}^{n_C}\mathbbm{1}_{\{q_j>u_C\}}}{n_C}\right)u_C^{1/\hat{\tau}_C}.
\end{align*}
This estimate of $\tau_C$ can exceed 1, so we prefer to use $\min(\hat{\tau}_C,1)$, with an appropriate change to $\hat{K}_C$. The Hill estimator for $\tau_C$ is consistent if $u_C\rightarrow\infty$, $\sum_j \mathbbm{1}_{\{q_j>u_C\}}\rightarrow\infty$ and $\sum_j\mathbbm{1}_{\{q_j>u_C\}}/n_C\rightarrow 0$; the assumption of $L_C(q)\sim K_C>0$ is not required for this. The second condition ensures that the number of points in $E_C$ with $Q > u_C$ goes to infinity, and since the expected number $n_C\text{pr}(Q>u_C\mid X^*\in E_C)\sim n_C K_C u_C^{-1/\tau_C}$, this entails $u_C=o(n_C^{\tau_C})$.

The method of \cite{Goix2015} produces empirical estimates of $\text{pr}(X\in E_{C}\mid R>r_{0})$, for $R=\|X\|_\infty$, and some value of $r_{0}$ within the range of observed values. These estimates are then assumed to hold for all $R>r_0$ and are used to approximate the limit. If the conditional probability $\text{pr}(X\in E_C\mid R>r)$ changes with $r>r_0$, \citeauthor{Goix2015}\ estimate this as a positive constant or as zero. In contrast, our semiparametric method allows us to estimate $\text{pr}(X^*\in E_{C}\mid Q>q)$, for all $q$ above a high threshold, via
\begin{align}
	\text{pr}(X^*\in E_C\mid Q>q) = \frac{\text{pr}(Q>q\mid X^*\in E_C)\text{pr}(X^*\in E_C)}{\sum\limits_{C'\in 2^D\setminus\emptyset}\text{pr}(Q>q\mid X^*\in E_{C'})\text{pr}(X^*\in E_{C'})},~~~\text{$C \in 2^{D}\setminus\emptyset$},
    \label{eqn:mainResult}
\end{align}
with $E_{C}$ as in \eqref{eqn:regions}. Our estimate of probability~\eqref{eqn:mainResult} varies continuously with $q$, with this variation being determined by the estimated values $\hat\tau_C$, for $C\in 2^D\setminus\emptyset$. In situations where sub-asymptotic dependence leads to many points in a region $E_C$, but $\mu(\mathbb{E}_C) = 0$ and $\hat{\tau}_C<1$, this extrapolation can be helpful in obtaining a better approximation to the limit. The relative merits of these differences to the approach of \citeauthor{Goix2015}, which are common to Methods~1~and~2, are illustrated in Sections~\ref{sec:simulation}~and~\ref{sec:data}.

The right-hand side of equation~\eqref{eqn:mainResult} consists of two types of component. We estimate terms of the form $\text{pr}(X^*\in E_C)$ empirically, and we estimate those of the form $\text{pr}(Q>q\mid X^*\in E_C)$ as in \eqref{eqn:HRVassumption} by replacing $K_C$ and $\tau_C$ by their estimates, and evaluating for some large choice of $q$, discussed in Section~\ref{sec:data}. This approach yields an estimate for the proportion of mass in each region. We denote the estimated vector of these proportions by $\hat{p} = (\hat{p}_{C}:C\in 2^{D}\setminus\emptyset)$. In order to obtain a sparse representation of the mass on the simplex, we follow \citeauthor{Goix2015}\ (\citeyear{Goix2016}, \citeyear{Goix2015}) and ignore any mass that has been detected which is considered to be negligible; see use of parameter $\pi$, below. A summary of our method is as follows. 

First, transform the data to standard Fr\'{e}chet margins, and for a choice of the tuning parameter $p$, apply transformation \eqref{eqn:truncation}. Then assign each transformed observation to a region $E_C$ as in \eqref{eqn:regions}, removing any all-zero points. For each region $E_C$ containing more than $m$ points, fit model~\eqref{eqn:HRVassumption} for a choice of threshold $u_C$, and estimate $\text{pr}(X^*\in E_C\mid Q>q)$ for a large value of $q$ by equation \eqref{eqn:mainResult}. Set $\text{pr}(X^*\in E_C\mid Q>q)=0$ in the remaining regions, denoting the resulting estimate by $\hat{p}_C$. Finally, if $\hat{p}_C < \pi$, for a choice of the tuning parameter $\pi$, set $\hat{p}_C$ to zero, renormalizing the resulting vector.

The parameter $m$ ensures there are enough points to estimate the parameters on each cone $\mathbb{E}_C$. In simulations, it was found not to have a significant effect on results, so we take $m=1$.

\subsection{Method 2: $\delta>0$}\label{subsec:method2}
An alternative to setting $\delta=0$ and partitioning the positive orthant using regions $E_C$, is to consider $\delta>0$ in the application of Assumption~\ref{assumption:HRV}, specifically $\text{pr}\left(X_\wedge^C>t,X_\vee^{D\setminus C}\leq t^\delta\right)\in {\rm{RV}}_{-1/\tau_C(\delta)}$. However, unlike with $\delta=0$, this does not lead directly to a univariate structure variable with tail index $1/\tau_C(\delta)$. We instead consider $\text{pr}\left\{X_\wedge^C>t,X_\vee^{D\setminus C}\leq \left(X_\wedge^C\right)^\delta\right\}=\text{pr}\left(X_\wedge^C >t, X\in\tilde{E}_C\right)$, with $\tilde{E}_C$ defined as
\begin{align*}
\tilde{E}_C = \left\{x\in\mathbb{E}:x_\vee^{D\setminus C} \leq \left(x_\wedge^C\right)^\delta \right\},~|C|<d;~~~~~	\tilde{E}_D = \mathbb{E}\setminus \bigcup_{C\in 2^D\setminus\emptyset:|C|<d}\tilde{E}_C,
\end{align*}
for each $C\in 2^D\setminus\emptyset$. We denote the corresponding tail index as $1/\tilde{\tau}_C(\delta)$, and assume that
\begin{align*}
\text{pr}\left(X_\wedge^C >q, X\in\tilde{E}_C\right)\in{\rm{RV}}_{-1/\tilde\tau_C(\delta)},~~~~q\rightarrow\infty.
\end{align*}
Analogously to equation~\eqref{eqn:HRVassumption} of Method~1, for each region $\tilde{E}_C$, we assume the model
\begin{align}
\text{pr}(X_\wedge^C >q\mid X\in\tilde{E}_C) = K_C q^{~-1/\tilde{\tau}_{C}(\delta)},~~~q>u_{C},
\label{eqn:HRV}
\end{align}
for some large threshold $u_C$, where estimates of $K_C$ and $\tilde{\tau}_C(\delta)$ are again obtained by maximizing a censored likelihood. In the Supplementary Material, we examine estimates of $\tilde{\tau}_C(\delta)$, which we find to reasonably approximate the true values of $\tau_C(\delta)$. This indicates that the indices $\tilde{\tau}_C(\delta)$ provide useful information about $\tau_C(\delta)$. We note that the regions $\tilde{E}_C$ are not disjoint. Supposing we have observations $x_1,\dots,x_n$, we obtain an empirical estimate of $\text{pr}(X\in\tilde{E}_{C})$ using
\begin{align}
\frac{1}{n}  \sum\limits_{j=1}^{n} \frac{\mathbbm{1}_{\{x_j\in\tilde{E}_C\}}}{\sum_{C'\in 2^D\setminus\emptyset}\mathbbm{1}_{\{x_j\in \tilde{E}_{C'} \}}},
\label{eqn:weightedProb}
\end{align}
so that the contribution of each observation sums to one. Combining~\eqref{eqn:HRV}~and~\eqref{eqn:weightedProb}, we estimate
\begin{align}
	\text{pr}\left(X_\wedge^C >q, X\in\tilde{E}_C\right) = \text{pr}(X_\wedge^C >q\mid X\in\tilde{E}_C)\text{pr}(X\in\tilde{E}_C),~~~~~\text{$C \in 2^{D}\setminus\emptyset$,}
\label{eqn:mainResult2}
\end{align}  
for some large $q$. To estimate the proportion of extremal mass associated with each cone $\mathbb{E}_C$, we consider probability~\eqref{eqn:mainResult2} for a given $\tilde{E}_C$ divided by the sum over all such probabilities, corresponding to $\tilde{E}_{C'}$, $C'\in 2^D\setminus\emptyset$. As in Method~1, the result is evaluated at a high threshold $q$, and mass estimated to be below the threshold $\pi$ is removed.

\section{Simulation study}\label{sec:simulation}
\subsection{Overview and metrics}\label{subsec:simulationIntro}
We present simulations to demonstrate Methods~1~and~2, and compare them with the approach of \cite{Goix2015}. Here, we consider a max-mixture distribution involving Gaussian and extreme~value logistic distributions, described in equation~\eqref{eqn:mixtureDistribution}. In the Supplementary Material we present results for a special case of this, the asymmetric logistic distribution \citep{Tawn1990}, that is used by \cite{Goix2015} to assess the performance of their methods. The key difference between these two distributions is that the Gaussian components in the max-mixture model lead to sub-asymptotic dependence, in contrast to independence, on certain sub-cones. Hence, distinguishing between cones $\mathbb{E}_C$ with and without limiting mass is a more difficult task for our max-mixture distribution. For the classes of model we consider, it is possible to calculate the proportion of extremal mass on the various cones analytically, allowing us to compare our estimates to the true distribution of mass using the Hellinger distance.

When incorporating a cut-off for sparse representation of the measure $\mu$, as mentioned in Section~\ref{subsec:method1}, the methods can be viewed as classification techniques. Plotting receiver operating characteristic curves is a common method for testing the efficacy of classifiers \citep{Hastie2009}. To obtain such curves, the false positive rate of a method is plotted against the true positive rate, as some parameter of the method varies. In our case, the false positive rate is the proportion of cones $\mathbb{E}_C$ incorrectly detected as having mass, while the true positive rate is the proportion of correctly detected cones. To obtain our curves, we vary the threshold, $\pi$, above which estimated mass is considered non-negligible. For $\pi=0$, all cones will be included in the estimated dependence structure, leading to the true and false positive rates both being 1, while $\pi=1$ includes none of the cones, so both equal 0. A perfect result for a given data set and method would be a false positive rate of 0 and true positive rate of 1: the closer the curve is to the point $(0,1)$, the better the method. This is often quantified using the area under the curve, with values closer to 1 corresponding to better methods.

Let $p = \left(p_{C}; C\in 2^{D}\setminus\emptyset\right)$ denote the true proportion of mass on each cone, and denote its estimate by $\hat{p}$. The Hellinger distance between $p$ and $\hat{p}$,
\begin{align}
	\HD(p,\hat{p}) = \frac{1}{\surd 2}\left\{\sum_{C\in2^{D}\setminus\emptyset} \Big(p_{C}^{1/2} - \hat{p}_{C}^{1/2}\Big)^{2}\right\}^{1/2},
\label{eqn:hellinger}
\end{align}
is used to determine the precision of the estimated proportions. In particular, $\HD(p,\hat{p})\in [0,1]$, and equals 0 if and only if $p = \hat{p}$. The closer $\HD(p,\hat{p})$ is to 0, the better $p$ is estimated by $\hat{p}$. Errors on small proportions are penalized more heavily than errors on large proportions. A small positive mass on a region, estimated as zero, will incur a relatively heavy penalty.

\subsection{Max-mixture distribution}\label{subsec:maxMixture}
\cite{Segers2012} shows how to construct distributions that place extremal mass on different combinations of cones. Here, we take a different approach by considering max-mixture models with asymptotic and sub-asymptotic dependence in different cones. This can be achieved by using a mixture of extreme value logistic and multivariate Gaussian copulas, a particular example of which we consider here.

Let $Z_{C}=\big( Z_{i,C}:i\in C \big)$ be a $|C|$-dimensional random vector with standard Fr\'{e}chet marginal distributions, and $\{Z_{C}:C\in 2^{D}\setminus\emptyset\}$ be independent random vectors. Define the vector $X=(X_{1},\dots,X_{d})$ with components
\begin{align}
	X_{i} = \max_{C\in 2^{D}\setminus\emptyset:i\in C} \left(\theta_{i,C}Z_{i,C}\right), ~~~~ \theta_{i,C}\in[0,1],~~~~\sum\limits_{C\in 2^{D}\setminus\emptyset:i\in C}\theta_{i,C}=1,
\label{eqn:mixtureDistribution}
\end{align}
for $i=1,\dots,d$. The constraints on $\theta_{i,C}$ ensure that $X$ also has standard Fr\'{e}chet margins. The random vector $Z_{C}$ may exhibit asymptotic dependence, in which case mass will be placed on the cone $\mathbb{E}_C$, or it may exhibit asymptotic independence, in which case mass will be placed on the cones $\mathbb{E}_i$, $i\in C$.

Here, we consider one particular five-dimensional example. We define $Z_{1,2}$ and $Z_{4,5}$ to have bivariate Gaussian copulas with correlation parameter $\rho$, and $Z_{1,2,3}$, $Z_{3,4,5}$ and $Z_{1,2,3,4,5}$ to have three-dimensional and five-dimensional extreme value logistic copulas with dependence parameter $\alpha$. The bivariate Gaussian distribution is asymptotically independent with sub-asymptotic dependence, while the logistic distribution is asymptotically dependent for $\alpha \in (0,1)$. As such, the cones with mass resulting from this construction are $\mathbb{E}_1, \mathbb{E}_2, \mathbb{E}_4, \mathbb{E}_5, \mathbb{E}_{1,2,3}, \mathbb{E}_{3,4,5}$ and $\mathbb{E}_{1,2,3,4,5}$. The Gaussian components mean that cones $\mathbb{E}_{1,2}$ and $\mathbb{E}_{4,5}$ have no mass asymptotically, but the parameter $\rho$ controls the decay rate of the mass. We assign equal mass to each of the seven charged cones by setting
\begin{align*}
	\theta_{1,2} =&\left(5,5\right)/7,~~\theta_{4,5}=\left(5,5\right)/7,\\
	\theta_{1,2,3}=\left(1,1,3\right)/7,~~\theta_{3,4,5}&=\left(3,1,1\right)/7,~~\theta_{1,2,3,4,5}=\left(1,1,1,1,1\right)/7.
\end{align*}
In this model, the cones with mass are fixed, in contrast to the asymmetric logistic examples in the Supplementary Material, where following \cite{Goix2015}, they are chosen at random over different simulation runs. Setting $\rho=0$ in this max-mixture distribution gives an asymmetric logistic model.

\begin{figure}[!htbp]
\begin{center}
\includegraphics[width=\textwidth]{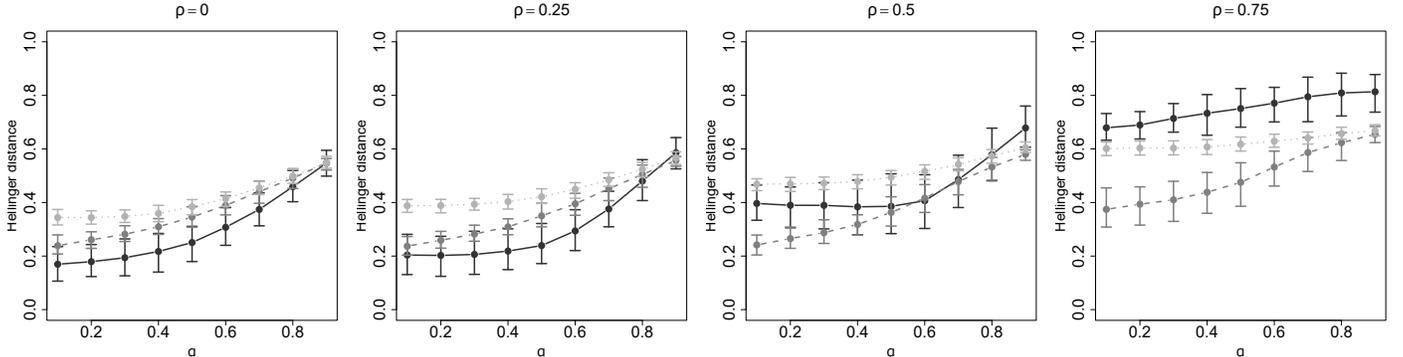}
\caption{Mean Hellinger distance, 0.05 and 0.95 quantiles over 100 simulations. Method~1: solid lines; Method~2: dashed lines; \citeauthor{Goix2015}: dotted lines.}
\label{fig:mixtureSimulations5D}
\end{center}
\end{figure}

Each setting in the simulation study is repeated 100 times, taking samples of size 10,000. In Method~1, we set $p=0.5$, $u_C$ to be the 0.75 quantile of observed $Q$ values in region $E_C$ for each $C\in 2^D\setminus\emptyset$, and the value of $q$ for which we estimate $\text{pr}(X^*\in E_{C}\mid Q>q)$ to be the 0.9999 quantile of all observed $Q$ values. In Method~2, we set $\delta=0.5$, each threshold $u_{C}$ to be the 0.85 quantile of observed $X_\wedge^C$ values in region $\tilde{E}_C$, and the extrapolation level $q$ to be the 0.9999 quantile of observed values of $X$. The parameters in the method of \cite{Goix2015} are chosen to be $(\epsilon,k) = (0.1,n^{1/2})$, using notation from that paper. When calculating the Hellinger distances, we used $\pi=0.001$ as the value above which estimated mass is considered significant in all three methods. The tuning parameters are not optimized for individual data sets, but fixed at values that we have found to work well across a range of settings. In Section~\ref{subsec:parameterStability}, we discuss stability plots, which could be used as a guide as to which tuning parameter values may be sensible for a given set of data. In Section~\ref{sec:data}, we consider how the estimated extremal dependence structure changes as the tuning parameters vary for a particular data set, allowing us to further examine this mass stability and choose a reasonable value of $p$ in Method~1 or $\delta$ in Method~2.

\begin{table}[!htbp]
\resizebox{\textwidth}{!}{%
\centering
\begin{tabular}{|c|ccc|ccc|ccc|ccc|}
\hline
&\multicolumn{3}{c|}{$\rho=0$}&\multicolumn{3}{c|}{$\rho=0.25$}&\multicolumn{3}{c|}{$\rho=0.5$}&\multicolumn{3}{c|}{$\rho=0.75$}\\
$\alpha$& 0.25 & 0.5 & 0.75 & 0.25 & 0.5 & 0.75 & 0.25 & 0.5 & 0.75 & 0.25 & 0.5 & 0.75 \\
\hline
\citeauthor{Goix2015} & 100 (0.0) & 100 (0.0) & 98.0 (1.1) & 99.7 (0.4) & 99.8 (0.4) & 96.3 (1.4) & 92.3 (0.6) & 91.9 (0.5) & 90.1 (1.2) &  91.0 (1.0) & 90.1 (1.7) & 87.6 (1.2)\\
Method~1 			  & 100 (0.0) & 100 (0.1) & 97.7 (1.4) & 100 (0.1) & 99.9 (0.3) & 96.7 (1.2) & 97.3 (1.6) & 96.3 (1.9) & 91.5 (1.9) & 92.9 (1.0) & 90.0 (0.9) & 87.5 (0.2)\\
Method~2  			  & 100 (0.0) & 99.2 (0.7) & 96.0 (1.6) & 100 (0.1) & 98.9 (0.8) & 94.6 (1.8) & 99.5 (0.6) & 97.5 (1.1) & 92.7 (1.7) & 94.4 (1.9) & 92.9 (1.9) & 89.1 (2.0)\\
\hline 
\end{tabular}}
 \caption{Average area under the receiver operating characteristic curves, given as percentages, for 100 samples from a five-dimensional mixture of bivariate Gaussian and extreme value logistic distributions; the standard deviation of each result is given in brackets.}
   \label{table:AUCmix}
\end{table}

\begin{figure}[!htbp]
\begin{center}
\includegraphics[width=\textwidth]{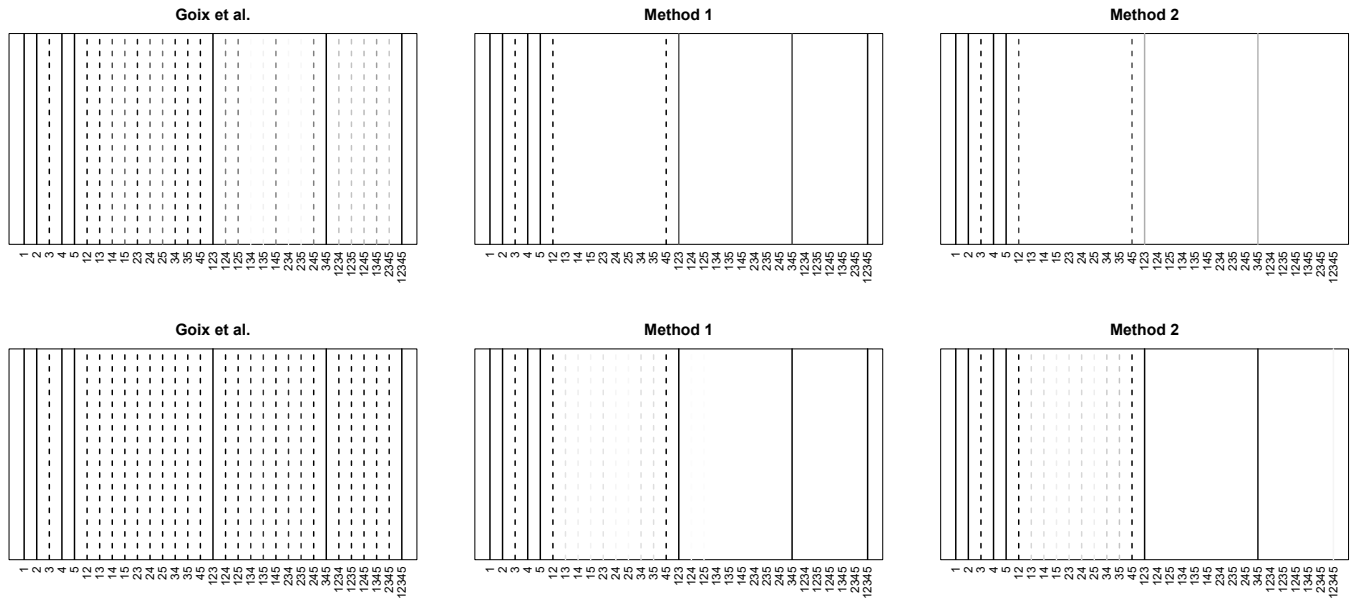}
\caption{Plots to show the number of times each cone is assigned mass greater than $\pi=0.01$ (top) and $\pi=0.001$ (bottom), for $(\alpha,\rho)=(0.75,0.5)$. Darker lines correspond to higher detection rates over 100 simulations. True cones with mass: solid lines; cones without mass; dashed lines.}
\label{fig:maxmixFaces}
\end{center}
\end{figure}

In Fig.~\ref{fig:mixtureSimulations5D}, we show the mean Hellinger distance achieved by each method for $\rho\in \{0,0.25,0.5,0.75\}$ and $\alpha \in [0.1,0.9]$. Results for the area under the receiver operating characteristic curves are provided in Table~\ref{table:AUCmix}. The performance of all three methods deteriorates as the value of the correlation parameter $\rho$, or the dependence parameter $\alpha$, increases. In the former case this is due to the stronger sub-asymptotic dependence on cones without extremal mass; in the latter case, larger values of $\alpha$ in logistic component $Z_C$ lead to larger values of $\tau_{\underline{C}}(\delta)$ for $\underline{C}\subset C$, so it is harder to determine which cones $\mathbb{E}_C$ truly contain extremal mass. In terms of the Hellinger distance, Method~1 is the most successful for $\rho=0,0.25$, although its performance deteriorates when there is stronger correlation in the Gaussian components. Method~2 yields the best results for $\rho=0.5,0.75$. In terms of estimating the proportion of extremal mass associated with each cone $\mathbb{E}_C$, at least one of our proposed methods is always more successful than \citeauthor{Goix2015}\ for this max-mixture model. The results in Table~\ref{table:AUCmix} reveal that all three methods are successful classifiers for low values of $\rho$ and $\alpha$. For $\alpha=0.75$ and $\rho=0,0.25$, Method~1 and the approach of \citeauthor{Goix2015}\ demonstrate similarly strong performance, while for $\rho=0.5,0.75$, Method~2 again provides the best results. 

As a further comparison of the methods, in Fig.~\ref{fig:maxmixFaces}, we investigate how often each cone $\mathbb{E}_C$ is detected as having mass above $\pi=0.01,0.001$ for the $(\alpha,\rho)=(0.75,0.5)$ case. For $\pi=0.001$, the approach of \citeauthor{Goix2015}\ places mass on around three times as many cones as Methods~1~and~2, and over twice as many for the $\pi=0.01$ case, so our methods provide sparser representations of the extremal mass that are both much closer to the truth. The reason for this difference is explained by the method of \citeauthor{Goix2015}\ assuming there is extremal mass on a cone $\mathbb{E}_C$ if $\text{pr}(X\in E_C\mid R>r_0)>\pi$, whereas we recognize that when $\hat\tau_C<1$ or $\hat\tau_C(\delta)<1$, non-limit mass can be on a cone at a finite threshold, but may progressively decrease to zero as the level of extremity of the vector variable is increased to infinity. When $\hat\tau_C=1$ or $\hat\tau_C(\delta)=1$, we estimate mass on cone $\mathbb{E}_C$ similarly to \citeauthor{Goix2015} As a consequence, our approach integrates information over the entire tail to estimate which cones have limit mass, as opposed to \citeauthor{Goix2015}, who use information only at a single quantile. We also observe from Fig.~\ref{fig:maxmixFaces} that Method~2 often fails to detect the cone corresponding to all five variables being large simultaneously, and places more mass on lower-dimensional cones, arising from the estimated values of $\tau_C(\delta)$. Method~1 also places mass on these lower-dimensional cones, but more often detects the true higher-dimensional cones with mass.

\subsection{Stability plots}\label{subsec:parameterStability}
One way to decide on reasonable tuning parameter values for a given set of data is via a parameter stability plot. Here, we outline how to construct such a plot for an example using the max-mixture distribution of Section~\ref{subsec:maxMixture} with Method~2, where our aim is to obtain a sensible range of values for the tuning parameter $\delta$, by considering the region of $\delta$ where the number of cones determined as having mass is stable.

\begin{figure}[!htbp]
\begin{center}
\includegraphics[width=0.8\textwidth]{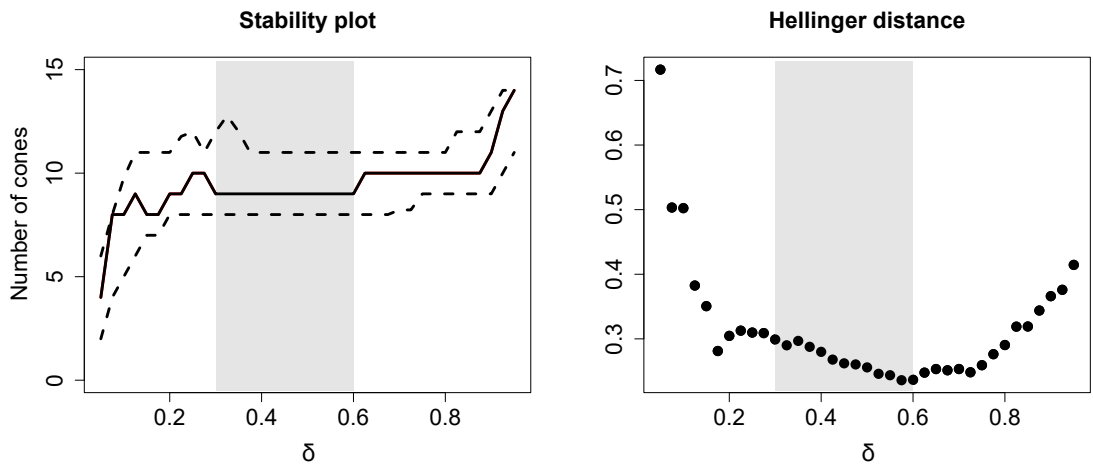}
\caption{Stability plot (left) for Method~2, with dashed lines showing a 95\% bootstrapped confidence interval for the number of cones $\mathbb{E}_C$ with mass, and a plot of the Hellinger distance (right) for each value of $\delta$. The shaded regions correspond to the stable range of tuning parameter values. Data were simulated from the max-mixture distribution of Section~\ref{subsec:maxMixture} with $n=10,000$, $\alpha=0.25$ and $\rho=0.25$.}
\label{fig:parameterStability5D}
\end{center}
\end{figure}

For $\delta\in\{0.05,0.075,\dots,0.95\}$, we use Method~2 to estimate the proportion of extremal mass on each cone, and find the number of cones whose estimated mass is greater than $\pi = 0.001$ in each case. The remaining parameters are fixed as in Section~\ref{subsec:maxMixture}. Figure~\ref{fig:parameterStability5D} shows the estimates of the number of cones, with a 95\% confidence interval constructed from 250 bootstrapped samples: these constitute our stability plot. Analogous plots can be created to choose $p$ in Method~1, or in each case to choose $\pi$. In practice, the choice of threshold $\pi$ should depend on the dimension of the data; this is not explored here.

The number of cones detected as having mass is most stable for values of $\delta$ between 0.3 and 0.6, indicated by the shaded regions in Fig.~\ref{fig:parameterStability5D}, suggesting values of $\delta$ in this range may be appropriate for this sample. The right-hand panel of Fig.~\ref{fig:parameterStability5D} shows the Hellinger distance corresponding to the set of estimated proportions obtained for each value of $\delta$. For this particular sample, although values of $\delta$ within the stable range slightly overestimate the number of cones with mass, the smallest Hellinger distance occurs for a value of $\delta$ within the stable range, and the Hellinger distance is reasonably consistent across these tuning parameter values. In practice, the true proportions on each cone $\mathbb{E}_C$ are unknown, so Hellinger plots cannot be constructed; the plot here supports the idea of using stability plots in choosing suitable tuning parameter values. There is no guarantee that stability plots will find the optimal tuning parameter values, but they do offer some insight into tuning parameter optimization. Consideration of the context of the problem may be useful in determining whether it is reasonable for extremal mass to be placed on particular combinations of cones, and such insight could facilitate the choice of different $p$ or $\delta$ values for different cones $\mathbb{E}_C$.

\section{River flow data}\label{sec:data}
We apply Methods~1~and~2 to daily mean river flow readings from 1980 to 2013, measured in cubic metres per second, at five gauging stations in the North West of England. These data are available from the Centre for Ecology and Hydrology at \texttt{nrfa.ceh.ac.uk} (\citeauthor{Morris1990}, \citeyear{Morris1990}, \citeyear{Morris1994}). Estimates of the extremal dependence structure of the flows could be used to aid model selection, or one could carry out density estimation on each cone $\mathbb{E}_C$ to give an overall model. The locations of the five gauges are shown in Fig.~\ref{fig:riverGaugeMaps}; the labels assigned to each location will be used to describe the dependence structures estimated in this section. Figure~\ref{fig:riverGaugeMaps} also illustrates the boundaries of the catchments associated with each gauge. These catchments demonstrate the areas from which surface water, usually as a result of precipitation, will drain to each gauge. The spatial dependence of river flow is studied by \cite{Keef2013} and \cite{Asadi2015}. As high river flow is mainly caused by heavy rainfall, we may observe extreme river flow readings at several locations simultaneously if they are affected by the same extreme weather event. Gauges with adjacent or overlapping catchments are expected to take their largest values simultaneously, with stronger dependence between gauges that are closer together.

\begin{figure}[!htbp]
\begin{center}
\includegraphics[width=0.7\textwidth]{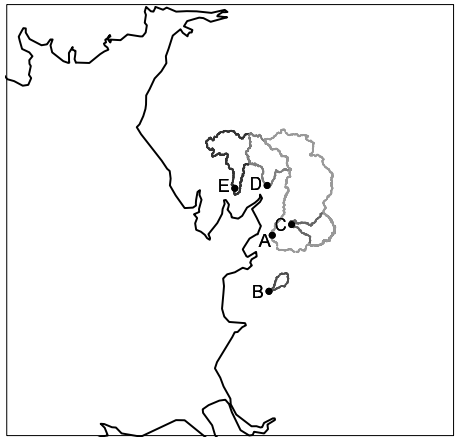}
\caption{Locations of the river flow gauges, labelled A to E, and corresponding catchment boundaries.}
\label{fig:riverGaugeMaps}
\end{center}
\end{figure}

Table~\ref{table:riverLocations1} shows the percentage of the extremal mass assigned to each cone for tuning parameter values $p\in\{0.7,0.725,\dots,0.975\}$ and $\delta\in\{0.2,0.25,\dots,0.75\}$. We set $\pi=0.01$ to be the threshold below which the proportion of mass is deemed negligible, and the extrapolation levels $q$ to be the 0.999 quantile of the observed $Q$ and $X$ values in Methods~1~and~2, respectively. Remaining parameters are fixed as in Section~\ref{subsec:maxMixture}. By observing how the estimated dependence structure changes over a range of tuning parameter values, we aim to find a `stable region' in which the results are most reliable. A further consideration is whether the tuning parameters give a feasible estimate of the extremal dependence structure. In particular, each variable should be represented on at least one cone, and moment constraint~\eqref{eqn:momentConstraint} should be taken into account. For Method~2, Table~\ref{table:riverLocations1} indicates that for $\delta\geq 0.45$, the cone corresponding to location E is assigned more than 20\% of the extremal mass, which is not possible due to the moment constraint. Feasible stable regions are demonstrated by the shaded regions in Table~\ref{table:riverLocations1}. For Method~2, one could also look for a value of $\delta$ that give estimates of $\tau_C(\delta)$ satisfying $\max_{C:C\supseteq i} \hat{\tau}_C(\delta)=1$, subject to estimation uncertainty, for every $i=1,\dots,d$.

\begin{table}[!htbp]
\begin{centering}
\resizebox{0.48\textwidth}{!}{%
\begin{tabular}{|c|ccccccccccccc|}
\hline
$p$ & B & C & D & E & AC & AD & BC & ABC & ACD & ADE & ABCD & ACDE & ABCDE \\  
\hline
  0.700 & 2   &  &  &  &  &  &  &  &  &  & 3   &  & 95   \\ 
  0.725 & 2   &  &  &  &  &  &  &  &  &  & 4   &  & 93   \\ 
  0.750 & 1   &  &  &  &  &  &  &  &  &  & 4   &  & 95   \\ 
  0.775 & 1   &  &  &  &  &  &  &  &  &  & 11   &  & 88   \\ 
  0.800 & 7   &  &  &  &  &  &  &  &  &  & 11   &  & 82   \\ 
  \rowcolor{grey}0.825 & 10   &  &  & 3   &  &  &  &  &  &  & 16   &  & 71   \\ 
  \rowcolor{grey}0.850 & 13   &  &  & 7   &  &  &  &  &  &  & 17   &  & 63   \\ 
  \rowcolor{grey}0.875 & 12   &  &  & 12   &  &  &  &  &  &  & 19   & 1   & 55   \\ 
  0.900 & 8   &  &  & 31   &  &  &  &  &  &  & 15   &  & 46   \\ 
  0.925 & 6   &  &  & 57   &  &  &  & 2   & 2   &  & 15   &  & 19   \\ 
  0.950 & 5   &  &  & 52   &  &  &  & 4   & 2   &  & 16   &  & 21   \\ 
  0.975 & 14   & 1   &  & 73   &  & 2   & 1   &  &  &  &  & 2   & 6   \\ 
  \hline
\end{tabular}}
\quad
\resizebox{0.48\textwidth}{!}{%
\begin{tabular}{|c|ccccccccccccc|}
\hline
$\delta$ & B & C & D & E & AC & AD & BC & ABC & ACD & ADE & ABCD & ACDE & ABCDE \\ \hline
  0.20 &  &  &  &  &  &  &  &  &  &  & 1 &  & 99 \\ 
  0.25 & 4 &  &  &  &  &  &  &  &  &  & 15 &  & 81 \\ 
  0.30 & 3 &  & 1 &  &  &  &  & 7 &  &  & 24 &  & 65 \\ 
  \rowcolor{grey}0.35 & 8 &  &  & 1 &  &  &  & 9 &  &  & 33 &  & 49 \\ 
  \rowcolor{grey}0.40 & 14 &  &  & 3 &  &  &  & 2 &  &  & 49 &  & 30 \\ 
  0.45 & 27 &  & 1 & 26 &  &  & 1 & 4 &  &  & 29 &  & 12 \\ 
  0.50 & 22 &  & 1 & 49 &  &  & 1 & 2 &  &  & 15 & 2 & 7 \\ 
  0.55 & 16 &  &  & 52 &  & 1 & 1 & 3 &  &  & 18 & 4 & 6 \\ 
  0.60 & 13 &  &  & 55 &  &  &  & 2 & 2 &  & 19 & 5 & 4 \\ 
  0.65 & 16 &  & 3 & 48 &  & 2 & 2 & 3 &  &  & 17 & 6 & 3 \\ 
  0.70 & 16 & 2 & 4 & 45 &  & 2 & 3 & 3 & 4 &  & 13 & 6 & 3 \\ 
  0.75 & 19 & 3 & 4 & 40 & 1 & 2 & 3 & 5 & 4 & 1 & 11 & 4 & 2 \\ 
\hline
\end{tabular}}
   \caption{The percentage of mass assigned to each sub-cone for varying values of the tuning parameters in Method~1 (left) and Method~2 (right). The grey regions demonstrate the feasible stable ranges.}
   \label{table:riverLocations1}
\end{centering}
\end{table}

Focusing on tuning parameter values within each of the stable regions in Table~\ref{table:riverLocations1}, Method~1 suggests the dependence structure to be \{B, E, ABCD, ABCDE\}, while Method~2 suggests \{B, E, ABC, ABCD, ABCDE\}. All the cones detected by Method~2 are either also detected by Method~1, or are neighbours of cones detected by Method~1, showing there is some agreement between the methods. If we had used a higher threshold for the negligible mass, say $\pi=0.1$, for tuning parameter values in the stable region, both methods would have detected the structure \{B, ABCD, ABCDE\}. We also investigated the behaviour of the methods using the 0.99 and 0.9999 quantiles for extrapolation level $q$. For both methods, the set of cones estimated as having mass was stable, but for the lower quantile Method~2 placed less mass on ABCDE, and there was more mass assigned to this cone at the higher quantile.

The subsets of locations detected as having simultaneously high river flows seem feasible when considering the geographic positions of the gauging stations. For instance, both methods suggest mass on ABCD; as station E lies towards the edge of the region under consideration, it is possible for weather events to affect only the other four locations. Both methods also suggest that locations B and E can experience high river flows individually; this seems reasonable as they lie at the edge of region we consider. The catchment of gauge C lies entirely within the catchment of gauge A. We observe that location C occurs with location A in the subsets of sites determined to take their largest values simultaneously, which may be a consequence of this nested structure.

\begin{table}[!htbp]
\centering
\resizebox{0.4\textwidth}{!}{%
\begin{tabular}{|c|cccccccccc|}
\hline
$p$ & A & B & C & D & AC & AD & BC & ABC & ACD & ABCD \\  
 \hline
  0.700 &  &  &  &  &  &  &  &  &  & 100 \\ 
  0.725 &  &  &  &  &  &  &  &  &  & 100 \\ 
  0.750 &  &  &  &  &  &  &  &  &  & 100 \\ 
  0.775 &  &  &  &  &  &  &  &  &  & 100 \\ 
  0.800 &  &  &  &  &  &  &  &  &  & 100 \\ 
  \rowcolor{grey}0.825 &  & 3 &  &  &  &  &  &  &  & 97 \\ 
  \rowcolor{grey}0.850 &  & 3 &  &  &  &  &  &  &  & 97 \\ 
  \rowcolor{grey}0.875 &  & 4 &  &  &  &  &  &  &  & 96 \\ 
  0.900 &  & 2 &  &  &  &  &  &  &  & 98 \\ 
  0.925 &  & 5 &  &  &  &  &  & 1 &  & 94 \\ 
  0.950 &  & 7 &  &  &  & 2 &  & 5 & 2 & 84 \\ 
  0.975 &  & 35 & 3 &  &  & 1 & 2 &  & 7 & 52 \\ 
 \hline
\end{tabular}}
\quad
\resizebox{0.4\textwidth}{!}{%
\begin{tabular}{|c|cccccccccc|}
\hline
$\delta$ & A & B & C & D & AC & AD & BC & ABC & ACD & ABCD \\ 
 \hline
  0.20 &  &  &  &  &  &  &  &  &  & 100 \\ 
  0.25 &  &  &  &  &  &  &  &  &  & 100 \\ 
  0.30 &  & 2 &  &  &  &  &  &  &  & 98 \\ 
  \rowcolor{grey}0.35 &  & 4 &  &  &  &  &  &  &  & 96 \\ 
  \rowcolor{grey}0.40 &  & 9 &  &  &  &  &  &  &  & 91 \\ 
  0.45 &  & 23 &  & 1 &  &  &  & 2 & 1 & 72 \\ 
  0.50 &  & 28 &  & 2 &  & 1 & 1 & 1 & 6 & 60 \\ 
  0.55 &  & 19 & 1 & 4 &  & 2 & 2 & 3 & 14 & 55 \\ 
  0.60 &  & 22 & 1 & 4 &  & 5 & 2 & 4 & 17 & 45 \\ 
  0.65 &  & 26 & 2 & 5 &  & 6 & 3 & 5 & 18 & 36 \\ 
  0.70 &  & 26 & 3 & 6 &  & 5 & 4 & 6 & 20 & 30 \\ 
  0.75 & 1 & 28 & 4 & 7 & 1 & 6 & 6 & 9 & 17 & 20 \\ 
\hline
\end{tabular}}
 \caption{Estimated percentage of extremal mass on to each sub-cone when considering locations A-D,  for varying values of the tuning parameters in Method~1 (left) and Method~2 (right).}
    \label{table:riverLocations4D}
\end{table}

To assess whether our methods are self-consistent across different dimensions, Table~\ref{table:riverLocations4D} shows similar results for locations A, B, C and D. We would expect the subsets of locations deemed to be simultaneously large to be the same as in Table~\ref{table:riverLocations1} if we ignore location E. Considering the same tuning parameter values as for Table~\ref{table:riverLocations1}, we see that the extremal dependence structures are estimated to be \{B, ABCD\} for both methods. For Method~1, this is the set of cones we would expect based on the five-dimensional results. For Method~2, we would also expect to detect the cone labelled ABC, although this was only assigned a relatively small proportion of the mass in the five-dimensional case.

Tables~\ref{table:riverLocations1}~and~\ref{table:riverLocations4D} demonstrate the importance of tuning parameter selection in Methods~1~and~2. As $p$ or $\delta$ increase, we are more likely to detect mass on the one-dimensional cones, or cones corresponding to subsets of the variables with low cardinality. Likewise, for low values of $p$ or $\delta$, we assign more extremal mass to the cone representing all variables being simultaneously extreme. In practice, we should consider the feasibility of the detected dependence structures, as well as the stability of the regions determined to have extremal mass as $p$ or $\delta$ vary. Our methods could be used to impose structure in more complete models for multivariate extremes. Even if a handful of different options look plausible with some variation in $p$ or $\delta$, this is still a huge reduction over the full set of possibilities.

\section*{Acknowledgement}
We gratefully acknowledge the support of the Engineering and Physical Sciences Research Council, through the EP/L015692/1 STOR-i centre for doctoral training, and fellowship EP/P002838/1. We acknowledge the National River Flow Archive Centre for Ecology and Hydrology for use of the river flow and catchment boundary data. We thank the referees and associate editor for their comments.

\section*{Supplementary Material}
Supplementary Material includes calculations of $\tau_C(\delta)$ for copulas in Table~\ref{table:tauExamples}, simulation results for estimates of $\tau_C(\delta)$ in Method~2, plots of simulation results for the max-mixture distribution of Section~\ref{subsec:maxMixture}, and additional simulation results for the asymmetric logistic model.

\appendix
\section*{Appendix}
\subsection*{Calculation of $\tau_{C}(\delta)$ for a bivariate extreme value distribution}\label{app:tau}
We determine the value of $\tau_C(\delta)$, defined in~\eqref{eqn:RVassumption1}, by establishing the index of regular variation of 
\begin{align*}
\text{pr}\left(X_i>t,i\in C; X_j<t^\delta,j\in D\setminus C\right).
\end{align*}
Here, we calculate $\tau_{1}(\delta)$, $\tau_{2}(\delta)$ and $\tau_{1,2}$ for a bivariate extreme value distribution, with distribution function given in~\eqref{eqn:bveDist}. The exponent measure $V$ can be written as
\begin{align*}
	V(x,y) &= \frac{2}{y}\int_{0}^{1}(1-w)dH(w) - \frac{2}{y}\int_{\frac{x}{x+y}}^{1}(1-w)h(w)dw + \frac{2}{x}\int_{\frac{x}{x+y}}^{1}wh(w)dw + \frac{2\theta_{1}}{x}.
\end{align*}

To study $\tau_{1}(\delta)$, suppose that $h(w)\sim c_{1}(1-w)^{s_{1}}$ as $w\rightarrow 1$, for $s_{1}>-1$. For $x\rightarrow\infty$ and $y=o(x)$, applying Karamata's theorem (\citeauthor{Resnick2010}, \citeyear{Resnick2010}, Theorem~2.1), we have
\begin{align*}
	V(x,y) &= \frac{1}{y} - \frac{2c_{1}}{y(s_{1}+2)}\left(\frac{y}{x+y}\right)^{s_{1}+2} \{1+o(1)\}+ \frac{2c_{1}}{x(s_{1}+1)}\left(\frac{y}{x+y}\right)^{s_{1}+1} \{1+o(1)\} + \frac{2\theta_{1}}{x}\\
		&= \frac{1}{y} + 2c_{1}\left(\frac{y}{x+y}\right)^{s_{1}+1}\left\{\frac{1}{x(s_{1}+1)}-\frac{1}{(s_{1}+2)(x+y)}\right\}  \{1+o(1)\}+ \frac{2\theta_{1}}{x}\\
		&=\frac{1}{y} + \frac{2c_{1}y^{s_{1}+1}x^{-(s_{1}+2)}}{(s_{1}+1)(s_{1}+2)}\{1+o(1)\}+ \frac{2\theta_{1}}{x}.
\end{align*}
By this result,
\begin{align*}
\text{pr}&\left(X_1>t,X_2<t^\delta\right) = \text{pr}\left(X_2<t^\delta\right) - \text{pr}\left(X_1<t,X_2<t^\delta\right)=\exp\left(-t^{-\delta}\right) - \exp\left\{-V\left(t,t^\delta\right)\right\}\nonumber\\
	&=\exp\left(-t^{-\delta}\right) - \exp\left\{-\frac{1}{t^\delta} - \frac{2c_{1}t^{\delta(s_{1}+1)}t^{-(s_{1}+2)}}{(s_{1}+1)(s_{1}+2)}\{1+o(1)\}- \frac{2\theta_{1}}{t}\right\}\nonumber\\
	&=\left\{1-t^{-\delta}+o\left(t^{-\delta}\right)\right\}\\
	&~~~~~~~\cdot\left(1 - \left[1- \frac{2c_{1}t^{\delta(s_{1}+1)}t^{-(s_{1}+2)}}{(s_{1}+1)(s_{1}+2)} + o\left\{t^{\delta(s_{1}+1)-(s_1+2)}\right\}\right]\left\{1-2\theta_1t^{-1}+o\left(t^{-1}\right)\right\}\right)\nonumber\\
	&=\left\{2\theta_1t^{-1} + \frac{2c_{1}t^{\delta(s_{1}+1)-(s_{1}+2)}}{(s_{1}+1)(s_{1}+2)}\right\}\left\{1+o(1)\right\}.
\end{align*}
If $\theta_{1}>0$, i.e., the spectral measure places mass on $\{1\}$, we see that $\text{pr}\left(X_1>t,X_2<t^\delta\right)\sim 2\theta_1t^{-1}$ as $t\rightarrow\infty$, hence $\tau_{1}(\delta)=1$ for all $\delta\in[0,1]$. If $\theta_{1}=0$, we have $\tau_{1}(\delta)=\{(s_{1}+2)-\delta(s_{1}+1)\}^{-1}$, which increases from $(s_1+2)^{-1}$ at $\delta=0$ to 1 at $\delta=1$. By similar calculations, if $h(w)\sim c_{2}w^{s_{2}}$ as $w\rightarrow 0$ for $s_{2}>-1$, we have $\tau_{2}(\delta)=1$ if $\theta_{2}>0$, and $\tau_{2}(\delta)=\{(s_{2}+2)-\delta(s_2+1)\}^{-1}$ otherwise. Since $\tau_{1,2}=\eta_{1,2}$, we have $\tau_{1,2}=1$ if $\theta_1+\theta_2<1$, and $\tau_{1,2}=1/2$ if $\theta_1+\theta_2=1$.

\newpage
\chapter{}
\begin{center}
{\bf\Large Supplementary Material}
\end{center}

\appendix
\setcounter{section}{0}

\section{Calculation of $\tau_{C}(\delta)$}
\subsection{Overview}
In the Appendix, we derived $\tau_{1}(\delta)$, $\tau_{2}(\delta)$ and $\tau_{1,2}$ for a particular subclass of bivariate extreme value distribution. Here, we present further calculations of $\tau_{C}(\delta)$ for several trivariate copula models.

In general, there are two cases to consider: $\delta=0$ and $\delta>0$. For many models, the asymptotic relations we study will differ by a constant in these two cases, while the tail index remains the same. For this reason, we focus on $\delta>0$, and present $\delta=0$ calculations separately only when the slowly varying function is no longer a constant, but instead varies with $t$.


\subsection{Independence copula}
We begin by considering the case where all three variables $X_1,X_2,X_3$ are independent. To calculate the value of $\tau_C(\delta)$, we need to determine the index of regular variation of 
\[
	\text{pr}\left(X_i>t, i\in C ; X_j<t^\delta,j\in D\setminus C\right).
\]
In the independence case, this is equivalent to
\[
\text{pr}\left(X_i>t\right)^{|C|} \text{pr}\left(X_i<t^\delta\right)^{|D\setminus C|} = \left(1-e^{-1/t}\right)^{|C|}\left(e^{-1/t^\delta}\right)^{|D\setminus C|}\sim t^{-|C|},
\]
so that $\tau_C(\delta)=1/|C|$, which does not depend on the value of $\delta$. That is, $\tau_1(\delta)=\tau_2(\delta)=\tau_3(\delta)=1$, $\tau_{1,2}(\delta)=\tau_{1,3}(\delta)=\tau_{2,3}(\delta)=1/2$, and $\tau_{1,2,3}=1/3$.

\subsection{Trivariate logistic distribution}
The trivariate extreme value logistic distribution belongs to the class of trivariate extreme value distributions. The exponent measure of the logistic distribution has the form
\begin{align}
	V(x,y,z) = \left( x^{-1/\alpha} + y^{-1/\alpha} + z^{-1/\alpha} \right)^{\alpha},
\label{eqn:logisticExponent}
\end{align}
for $\alpha\in(0,1]$. Since $\alpha=1$ corresponds to the independence case, we restrict our calculations to $\alpha\in(0,1)$. This distribution exhibits asymptotic dependence, with all limiting mass on $\mathbb{E}_{1,2,3}$. Since $\tau_{1,2,3}=\eta_{1,2,3}=1$, our interest lies with the values of $\tau_{C}(\delta)$ for $|C|=1$ and $|C|=2$, and we consider each of these in turn.


\noindent{\bf $|C|=1$: $\tau_1(\delta)$, $\tau_2(\delta)$, $\tau_3(\delta)$.}
In a similar approach to the bivariate case, we calculate $\tau_{1}(\delta)$ by considering
\begin{align*}
	\text{pr}&\left( X_1 > t , X_2<t^\delta , X_3<t^\delta   \right) =\text{pr}\left(X_2<t^\delta,X_3<t^\delta\right)- \text{pr}\left(X_1<t,X_2<t^\delta,X_3<t^\delta\right)\\
	&=\exp\left(-2^{\alpha}t^{-\delta}\right) - \exp\left[- 2^\alpha t^{-\delta}\left\{1+2^{-1}t^{(\delta-1)/\alpha}\right\}^\alpha\right]\\
&=1-2^{\alpha}t^{-\delta}+O(t^{-2\delta}) - \exp\left(- 2^\alpha t^{-\delta}\left[1+2^{-1}\alpha t^{(\delta-1)/\alpha} + O\left\{t^{2(\delta-1)/\alpha}\right\}\right]\right)\\
	&= 2^{\alpha-1}\alpha t^{(\delta-1-\alpha\delta)/\alpha}+O(t^{-2\delta}) +  O\left\{t^{(2\delta-2-\alpha\delta)/\alpha}\right\}\sim 2^{\alpha-1}\alpha t^{(\delta-1-\alpha\delta)/\alpha},
	\end{align*}
yielding $\tau_1(\delta) = \alpha/(1+\alpha\delta-\delta)$. By similar calculations, we have $\tau_2(\delta) = \tau_3(\delta)=\alpha/(1+\alpha\delta-\delta)$, which increase from $\alpha<1$ at $\delta=0$ to 1 at $\delta=1$.


\noindent{\bf $|C|=2$: $\tau_{1,2}(\delta)$, $\tau_{1,3}(\delta)$, $\tau_{2,3}(\delta)$.}
We carry out a similar calculation to find the value of $\tau_{1,2}(\delta)$. Here, we have
\begin{align*}
	\text{pr}&\left( X_1 > t , X_2>t , X_3<t^\delta   \right)\\
	&=\text{pr}\left(X_{3}<t^\delta\right) - \text{pr}\left(X_{1}<t,X_{3}<t^\delta\right) - \text{pr}\left(X_2<t,X_3<t^\delta\right)\\
	&\hspace{2.5cm}+\text{pr}\left(X_1<t,X_2<t,X_3<t^\delta\right)\\
	&=\exp\left(-t^{-\delta}\right) - 2\exp\left\{-\left(t^{-1/\alpha}+t^{-\delta/\alpha}\right)^\alpha\right\}+\exp\left\{-\left(2t^{-1/\alpha}+t^{-\delta/\alpha}\right)^\alpha\right\}\\
	&=\exp\left(-t^{-\delta}\right)\bigg\{1 - 2\exp\left(-t^{-\delta}\left[\alpha t^{(\delta-1)/\alpha} +\frac{1}{2}\alpha(\alpha-1)t^{2(\delta-1)/\alpha} + O\left\{t^{3(\delta-1)/\alpha}\right\}\right]\right) \\
	&\hspace{2.5cm}+ \exp\left(-t^{-\delta}\left[2\alpha t^{(\delta-1)/\alpha} + 2\alpha(\alpha-1)t^{2(\delta-1)/\alpha} +O\left\{t^{3(\delta-1)/\alpha}\right\}\right]\right)\bigg\}
\\
	&=\left\{1-t^{-\delta}+O\left(t^{-2\delta}\right)\right\}\left[\alpha(1-\alpha)t^{(2\delta-2-\alpha\delta)\alpha} + O\left\{t^{(3\delta-3-\alpha\delta)/\alpha}\right\}\right]\\
	&\sim \alpha\left(1-\alpha\right)t^{\left(2\delta-2-\alpha\delta\right)/\alpha}.
\end{align*}
This implies that $\tau_{1,2}(\delta)=\alpha/\left(2+\alpha\delta-2\delta\right)$, which varies from $\alpha/2$ at $\delta=0$ to 1 at $\delta=1$. Similarly, we have $\tau_{1,3}(\delta)=\tau_{2,3}(\delta)=\alpha/\left(2+\alpha\delta-2\delta\right)$.

These calculations reveal different indices of regular variation on cones $\mathbb{E}_C$ with $|C|=1,2$ in the trivariate logistic case.


\subsection{Trivariate distribution with extremal mass on one vertex and one edge}
Now we consider a trivariate example where the extremal mass is placed on one cone $\mathbb{E}_C$ with $|C|=1$, and another with $|C|=2$. This can be achieved by taking $(X_{1},X_{2})$ to have a bivariate extreme value logistic distribution, and $X_{3}$ to be a standard Fr\'{e}chet random variable independent of $(X_{1},X_{2})$. The exponent measure in this case has the form
\[
	V(x,y,z) = \left( x^{-1/\alpha} + y^{-1/\alpha} \right)^{\alpha} + z^{-1},~~~\text{$\alpha \in (0,1)$.}
\]


\noindent{\bf $|C|=1$: $\tau_{1}(\delta)$,  $\tau_{2}(\delta)$,  $\tau_3(\delta)$.}
We first consider the index of regular variation on the cone $\mathbb{E}_1$. Following a similar procedure to previously, and exploiting the independence of $(X_{1},X_{2})$ and $X_{3}$, we have
\begin{align*}
\text{pr}\left( X_1 > t , X_2<t^\delta , X_3<t^\delta   \right) &=\text{pr}\left(X_3<t^\delta\right)\left\{\text{pr}\left(X_2<t^\delta\right)-\text{pr}\left(X_1<t,X_2<t^\delta\right)\right\}\\
	&\hspace{-3cm} =\exp\left(-t^{-\delta}\right)\left(\exp\left(-t^{-\delta}\right)-\exp\left[-t^{-\delta}\left\{1+t^{(\delta-1)/\alpha}\right\}^\alpha \right]\right)\\
	&\hspace{-3cm} =\exp\left(-2t^{-\delta}\right)\left\{1-\exp\left(-t^{-\delta}\left[\alpha t^{(\delta-1)/\alpha}+O\left\{t^{2(\delta-1)/\alpha}\right\}\right] \right)\right\}\\
	&\hspace{-3cm}=\left\{1-2t^{-\delta}+O\left(t^{-2\delta}\right)\right\}\left[\alpha t^{(\delta-1-\alpha\delta)/\alpha} + O\left\{t^{(2\delta-2-\alpha\delta)/\alpha}\right\}\right]\sim \alpha t^{(\delta-1-\alpha\delta)/\alpha},
\end{align*}
revealing that $\tau_{1}(\delta)=\alpha/(1+\alpha\delta-\delta)$. By similar calculations, $\tau_{2}(\delta)=\alpha/(1+\alpha\delta-\delta)$. For the cone $\mathbb{E}_3$, we have $\tau_3(\delta)=1$, since
\begin{align*}
\text{pr}\left( X_1 < t^\delta , X_2<t^\delta , X_3>t\right) &= \text{pr}\left(X_3>t\right)\text{pr}\left( X_1 < t^\delta , X_2<t^\delta\right)\\
	&\hspace{-3cm}=\left(1-e^{-1/t}\right)\exp\left\{-\left(t^{-\delta/\alpha}+t^{-\delta/\alpha}\right)^\alpha\right\}=\left(1-e^{-1/t}\right)\exp\left(-2^\alpha t^{-\delta}\right)\\
	&\hspace{-3cm}=\left\{t^{-1}+O\left(t^{-2}\right)\right\}\left\{1-2^\alpha t^{-\delta}+O\left(t^{-2\delta}\right)\right\}\sim t^{-1}.
\end{align*}


\noindent{\bf $|C|=2$: $\tau_{1,2}(\delta)$, $\tau_{1,3}(\delta)$, $\tau_{2,3}(\delta)$.}
We begin by showing that $\tau_{1,2}(\delta)=1$. We have
\begin{align*}
\text{pr}\left( X_1 > t, X_2>t , X_3<t^\delta\right) &\\
&\hspace{-3cm}= \text{pr}\left(X_3<t^\delta\right)\left\{1 - \text{pr}\left(X_1<t\right) - \text{pr}\left(X_2<t\right) + \text{pr}\left( X_1 < t , X_2<t\right)\right\}\\
&\hspace{-3cm}=\exp\left(-t^{-\delta}\right)\left\{1-2\exp\left(-1/t\right)+\exp\left(-2^\alpha t^{-1}\right)\right\}\\
&\hspace{-3cm}=\left\{1-t^{-\delta}+O\left(t^{-2\delta}\right)\right\}\left\{\left(2-2^\alpha\right)t^{-1}+O(t^{-2})\right\}\sim\left(2-2^\alpha\right)t^{-1}.
\end{align*}
Next, we consider the cone $\mathbb{E}_{1,3}$. In this case, we have
\begin{align*}
	\text{pr}&\left( X_1> t , X_2<t^\delta , X_3>t \right)=\text{pr}\left(X_3 > t \right)\left\{\text{pr}\left( X_2 < t^\delta\right)  - \text{pr}\left( X_1 < t, X_2 <t^\delta\right)\right\}\\
	&=\left\{1-\exp\left(-1/t\right)\right\}\left[\exp\left(-t^{-\delta}\right) -\exp\left\{-\left(t^{-1/\alpha}+t^{-\delta/\alpha}\right)^\alpha\right\} \right]\\
	&=\left\{t^{-1}+O\left(t^{-2}\right)\right\}\exp\left(-t^{-\delta}\right)\left(1 - \exp\left[ -\alpha t^{(\delta-1-\alpha\delta)/\alpha} + O\left\{t^{(2\delta-2-\alpha\delta)/\alpha}\right\} \right]\right)\\	
	&=\left\{t^{-1}+O\left(t^{-2}\right)\right\}\left\{1-t^{-\delta}+O\left(t^{-2\delta}\right)\right\}\left[\alpha t^{(\delta-1-\alpha\delta)/\alpha} + O\left\{t^{(2\delta-2-\alpha\delta)/\alpha}\right\} \right]\\
	&\sim \alpha t^{(\delta-1-\alpha\delta-\alpha)/\alpha},
\end{align*}
so we have $\tau_{1,3}(\delta)=\alpha/(\alpha\delta+1+\alpha-\delta)$. Again by symmetry, $\tau_{2,3}(\delta)=\alpha/(\alpha\delta+1+\alpha-\delta)$. These indices vary from $\alpha/(1+\alpha)$ at $\delta=0$ to $1/2$ at $\delta=1$.


\noindent{\bf $|C|=3$: $\tau_{1,2,3}$.}
Finally, we consider the cone $\mathbb{E}_{1,2,3}$, where
\begin{align*}
	\text{pr}\left( X_1 > t, X_2>t, X_3>t\right) &\\
	&\hspace{-2cm}= \text{pr}\left(X_3>t\right)\left\{1-\text{pr}\left( X_1 < t\right) -\text{pr}\left( X_2 < t\right) + \text{pr}\left( X_1 < t, X_2 <t\right)\right\}\\
	&\hspace{-2cm}=\left\{1-\exp\left(-1/t\right)\right\}\left\{1-2\exp\left(-1/t\right)+\exp\left(-2^\alpha t^{-1}\right)\right\}\\
	&\hspace{-2cm}=\left\{t^{-1}+O\left(t^{-2}\right)\right\}\left\{\left(2-2^\alpha\right)t^{-1}+O\left(t^{-2}\right)\right\}\sim \left(2-2^\alpha\right)t^{-2},
\end{align*}
i.e., $\tau_{1,2,3}=1/2$.

\subsection{Trivariate inverted logistic distribution}
Next, we consider an inverted trivariate extreme value distribution, defined via its distribution function
\begin{align*}
	\text{pr}(X_{1}<x,X_{2}<y,X_{3}<z) &\\
	&\hspace{-3cm}= 1-\left\{ F_{X_1}(x') + F_{X_2}(y') + F_{X_3}(z') \right\} \\
	&\hspace{-3cm}+ \left\{F_{X_1,X_2}(x',y')+F_{X_1,X_3}(x',z')+F_{X_2,X_3}(y',z')\right\}-F_{X_1,X_2,X_3}(x',y',z'),
\end{align*}
where $F_{X_1,X_2,X_3}$ denotes the corresponding trivariate extreme value distribution function; $F_{X_1,X_2}$, $F_{X_1,X_3}$ and $F_{X_2,X_3}$ are the corresponding bivariate distribution functions; $F_{X_1}$, $F_{X_2}$ and $F_{X_3}$ are the marginal distributions of $X_1$, $X_2$ and $X_3$; and noting that $-\log\left(1-e^{-1/x}\right) = -\log\left\{x^{-1} + O(x^{-2})\right\}$,
\[
x'=-\frac{1}{\log\left(1-e^{-1/x}\right)}\sim\frac{1}{\log x},
\]
as $x\rightarrow\infty$, with $y'$, $z'$ defined analogously. In the case of the trivariate inverted logistic distribution, which we focus on here, $F_{X_1,X_2,X_3}(x,y,z)=\exp\left\{-V(x,y,z)\right\}$, for $V$ defined as in \eqref{eqn:logisticExponent}. The inverted logistic distribution exhibits asymptotic independence, placing all extremal mass on the cones $\mathbb{E}_C$ with $|C|=1$. We will show that $\tau_{1}(\delta)=\tau_{2}(\delta)=\tau_{3}(\delta)=1$ for this model, and then calculate $\tau_{1,2}(\delta), \tau_{1,3}(\delta), \tau_{2,3}(\delta)$ and $\tau_{1,2,3}$.

\noindent{\bf $|C|=1$: $\tau_{1}(\delta)$, $\tau_{2}(\delta)$, $\tau_{3}(\delta)$.}
To begin, we focus on calculating $\tau_{1}(\delta)$ for $\delta>0$, by considering
\begin{align*}
\text{pr}\left(X_1>t, X_2<t^\delta, X_3<t^\delta\right) &\\
&\hspace{-4cm}=\text{pr}\left(X_1>t\right) - \text{pr}\left(X_1>t,X_2>t^\delta\right)\\
&\hspace{-3cm}-\text{pr}\left(X_1>t,X_3>t^\delta\right)+\text{pr}\left(X_1>t,X_2>t^\delta,X_3>t^\delta\right)\\
&\hspace{-4cm}=\left(1-e^{-1/t}\right) - 2\exp\left(-\left[\left\{\log t + O\left(t^{-1}\right)\right\}^{1/\alpha} + \left\{\delta\log t + O\left(t^{-2\delta}\right)\right\}^{1/\alpha}\right]^\alpha\right)\\
&\hspace{-3cm}+ \exp\left(-\left[\left\{\log t + O\left(t^{-1}\right)\right\}^{1/\alpha} + 2\left\{\delta\log t + O\left(t^{-2\delta}\right)\right\}^{1/\alpha}\right]^\alpha\right)\\
&\hspace{-4cm}=\left\{t^{-1}+O\left(t^{-2}\right)\right\}- 2\exp\left[-\left(1+\delta^{1/\alpha}\right)^\alpha\log t + O\left\{(t\log t)^{-1}\right\}\right] \\
&\hspace{-3cm}+ \exp\left[-\left(1+2\delta^{1/\alpha}\right)^\alpha\log t+ O\left\{(t\log t)^{-1}\right\}\right]\\
&\hspace{-4cm}=\left\{t^{-1} - 2t^{-\left(1+\delta^{1/\alpha}\right)^\alpha} + t^{-\left(1+2\delta^{1/\alpha}\right)^\alpha}\right\}\left\{1+o(1)\right\}\sim t^{-1},
\end{align*}
so we have $\tau_1(\delta)=1$ for $\delta>0$.

\noindent For $\delta=0$, we consider variables $X_1^*,X_2^*,X_3^*$ defined via truncation~\eqref{eqn:truncation}, and study
\begin{align*}
&\text{pr}\left(X_{1}^{*}>t,X_{2}^{*}=0,X_{3}=0\right)=\text{pr}\left(X_{1}>t,X_{2}< -1/\log p,X_{3}< -1/\log p\right)\\
	&=\text{pr}\left(X_1>t\right) - \text{pr}\left(X_{1}>t,X_{2}>-1/\log p\right) - \text{pr}\left(X_{1}>t,X_{3}>-1/\log p\right)  \\
	&\hspace{1cm}+ \text{pr}\left(X_{1}>t,X_{2}>-1/\log p,X_{3}>-1/\log p\right)\\
	&= \left\{t^{-1} + O\left(t^{-2}\right)\right\} - 2\exp\left(-\left[\left\{\log t +O\left(t^{-1}\right)\right\}^{1/\alpha}+\left\{-\log(1-p)\right\}^{1/\alpha}\right]^\alpha\right)\\
	&\hspace{1cm}+\exp\left(-\left[\left\{\log t +O\left(t^{-1}\right)\right\}^{1/\alpha}+2\left\{-\log(1-p)\right\}^{1/\alpha}\right]^\alpha\right)\\
	&= \left\{t^{-1}+O\left(t^{-2}\right)\right\} - 2\exp\left\{-\left\{\log t +O\left(t^{-1}\right)\right\}	\left(1+\left[\frac{-\log(1-p)}{\left\{\log t +O\left(t^{-1}\right)\right\}}\right]^{1/\alpha}\right)^\alpha\right\}\\
	&\hspace{1cm}+\exp\left\{-\left\{\log t +O\left(t^{-1}\right)\right\}\left(1+2\left[\frac{-\log(1-p)}{\left\{\log t +O\left(t^{-1}\right)\right\}}\right]^{1/\alpha}\right)^\alpha\right\}\\
	&= \left\{t^{-1}+O\left(t^{-2}\right)\right\} \\
	&\hspace{1cm} - 2\exp\Bigg[-\left\{\log t +O\left(t^{-1}\right)\right\} -\alpha\frac{\left\{-\log(1-p)\right\}^{1/\alpha}}{\left\{\log t +O\left(t^{-1}\right)\right\}^{1/\alpha-1}} \\
	&\hspace{3cm}- \frac{\alpha(\alpha-1)}{2}\frac{\left\{-\log(1-p)\right\}^{2/\alpha}}{\left\{\log t +O\left(t^{-1}\right)\right\}^{2/\alpha-1}}+ o\left\{\left(\log t\right)^{1-2/\alpha}\right\}\Bigg]\\
	&\hspace{1cm}+\exp\Bigg[-\left\{\log t +O\left(t^{-1}\right)\right\} -2\alpha\frac{\left\{-\log(1-p)\right\}^{1/\alpha}}{\left\{\log t +O\left(t^{-1}\right)\right\}^{1/\alpha-1}} \\
	&\hspace{3cm}- 2\alpha(\alpha-1)\frac{\left\{-\log(1-p)\right\}^{2/\alpha}}{\left\{\log t +O\left(t^{-1}\right)\right\}^{2/\alpha-1}} + o\left\{\left(\log t\right)^{1-2/\alpha}\right\}\Bigg]\\			
	&=\left\{t^{-1}+O\left(t^{-2}\right)\right\} \\
	&\hspace{1cm}+t^{-1}\Bigg(
- 2\exp\Bigg[-\alpha\frac{\left\{-\log(1-p)\right\}^{1/\alpha}}{\left\{\log t +O\left(t^{-1}\right)\right\}^{1/\alpha-1}} \\
	&\hspace{3cm} +\frac{\alpha(1-\alpha)}{2}\frac{\left\{-\log(1-p)\right\}^{2/\alpha}}{\left\{\log t +O\left(t^{-1}\right)\right\}^{2/\alpha-1}} + o\left\{\left(\log t\right)^{1-2/\alpha}\right\} + O\left(t^{-1}\right)\Bigg]\\
	&\hspace{2cm}+\exp\Bigg[-2\alpha\frac{\left\{-\log(1-p)\right\}^{1/\alpha}}{\left\{\log t +O\left(t^{-1}\right)\right\}^{1/\alpha-1}}\\
	&\hspace{3cm} +2\alpha(1-\alpha)\frac{\left\{-\log(1-p)\right\}^{2/\alpha}}{\left\{\log t +O\left(t^{-1}\right)\right\}^{2/\alpha-1}} + o\left\{\left(\log t\right)^{1-2/\alpha}\right\}+O\left(t^{-1}\right)\Bigg]\Bigg)\\	
	&\sim \frac{t^{-1}}{\left(\log t\right)^{2/\alpha-1}} \left[\alpha(1-\alpha)\left\{-\log(1-p)\right\}^{2/\alpha}\right],
\end{align*}
as $t\rightarrow\infty$, which is regularly varying of order $-1$. As such, the index of regular variation is $\tau_1(\delta)=1$ for $\delta=0$. Combining these results, $\tau_1(\delta)=1$ for all $\delta\in[0,1]$. By symmetric arguments, $\tau_2(\delta)=\tau_3(\delta)=1$ for all $\delta\in[0,1]$.

{\noindent{\bf $|C|=2$: $\tau_{1,2}(\delta)$, $\tau_{1,3}(\delta)$, $\tau_{2,3}(\delta)$.}
We first consider the cone $\mathbb{E}_{1,2}$. For $\delta>0$, we have
\begin{align*}
\text{pr}&\left(X_1>t,X_2>t,X_3<t^\delta\right) = \text{pr}\left(X_1>t,X_2>t\right) -\text{pr}\left(X_1>t,X_2>t,X_3>t^\delta\right)\\
	&=\exp\left[-2^\alpha\left\{-\log\left(1-e^{-1/t}\right)\right\}\right] \\
	& \hspace{1cm}-\exp\left(-\left[2\left\{-\log\left(1-e^{-1/t}\right)\right\}^{1/\alpha} + \left\{-\log\left(1-e^{-1/t^\delta}\right)\right\}^{1/\alpha}\right]^\alpha\right)\\
	&=\left(1-e^{-1/t}\right)^{2^\alpha}-\exp\left(-\left[2\left\{\log t + O\left(t^{-1}\right)\right\}^{1/\alpha} + \left\{\delta\log t + O\left(t^{-2\delta}\right)\right\}^{1/\alpha}\right]^\alpha\right)\\
	&=t^{-2^\alpha} + O\left(t^{-1-2^\alpha}\right) - \exp\left[-\left(2+\delta^{1/\alpha}\right)^\alpha\log t+ O\left\{\left(t \log t \right)^{-1}\right\}\right] \sim t^{-2^\alpha}, 
\end{align*}
i.e., $\tau_{1,2}(\delta)=2^{-\alpha}$, $\delta>0$. For $\delta=0$, using similar arguments as for the $|C|=1$ case, we have
\begin{align*}
\text{pr}\left(X_{1}^{*}>t,X_{2}^{*}>t,X_{3}^{*}=0\right)	&=\text{pr}\left(X_{1}>t,X_{2}>t,X_{3}< -1/\log p\right)\\
	&\hspace{-4cm}=\text{pr}\left(X_{1}>t,X_{2}>t\right) - \text{pr}\left(X_{1}>t,X_{2}>t, X_{3}>-1/\log p\right) \\
	&\hspace{-4cm}=\exp\left[-2^\alpha \left\{\log t + O\left(t^{-1}\right)\right\}\right]\\
	&\hspace{-2cm}- \exp\left(-\left[ 2\left\{\log t + O\left(t^{-1}\right)\right\}^{1/\alpha} +\{-\log(1-p)\}^{1/\alpha} \right]^{\alpha}\right)\\
	&\hspace{-4cm}= \exp\left[-2^\alpha \left\{\log t + O\left(t^{-1}\right)\right\}\right]\\
	&\hspace{-2cm}\cdot\left(1-\exp\left[-\alpha 2^{\alpha-1}\frac{\{-\log(1-p)\}^{1/\alpha}}{\left\{\log t +O\left(t^{-1}\right)\right\}^{-1+1/\alpha}}+o\left\{\left(\log t\right)^{1-2/\alpha}\right\}\right]\right)\\
	&\hspace{-4cm}\sim \frac{t^{-2^{\alpha}}}{\left(\log t\right)^{-1+1/\alpha}}\alpha 2^{\alpha-1}\{-\log(1-p)\}^{1/\alpha},
\end{align*}
as $t\rightarrow\infty$. As such, the index of regular variation is $\tau_{1,2}(\delta)=2^{-\alpha}$ for all $\delta\in[0,1]$. By analogous arguments, we also have $\tau_{1,3}(\delta)=\tau_{2,3}(\delta)=2^{-\alpha}$, $\delta\in[0,1]$.

\noindent{\bf $|C|=3$: $\tau_{1,2,3}$.}
To calculate the index of regular variation for cone $\mathbb{E}_{1,2,3}$, we consider
\begin{align*}
	\text{pr}\left( X_1 > t, X_2>t, X_3>t\right) &=\exp\left\{\log\left(1-e^{-1/t}\right)V(1,1,1)\right\}\\
	& = \left(1-e^{-1/t}\right)^{V(1,1,1)} = \left\{1-1+t^{-1}+O(t^{-2})\right\}^{3^\alpha}\sim t^{-3^\alpha},
\end{align*}
so $\tau_{1,2,3}=3^{-\alpha}$. This corresponds to the known value of $\eta_{1,2,3}$ for the trivariate inverted logistic distribution, as $\eta_{1,2,3}=V(1,1,1)^{-1}=3^{-\alpha}$.


\subsection{Multivariate Gaussian distribution}
The multivariate Gaussian provides a further example of a distribution which asymptotically places all mass on cones $\mathbb{E}_C$ with $|C|=1$. In the case where $d=3$, for a multivariate Gaussian distribution with covariance matrix $\Sigma$, \cite{Nolde2014}, for example, shows that 
\[
	\eta_{1,2,3}=\left(\bmrem{1}_{3}^{T}\Sigma^{-1}\bmrem{1}_{3}\right)^{-1};~~~~\eta_{i,j}=\left(\bmrem{1}_{2}^T\Sigma_{i,j}^{-1}\bmrem{1}_{2}\right)^{-1},~~i<j \in\{1,2,3\},
\] 
where $\Sigma_{i,j}$ is the submatrix of $\Sigma$ corresponding to variables $i$ and $j$, and $\bmrem{1}_{d}\in\mathbb{R}^{d}$ is a vector of 1s.

The covariance matrix $\Sigma$ may be written as  
\[
\Sigma = \begin{bmatrix}
   	    1     & \rho_{12} & \rho_{13} \\
    \rho_{12} &     1     & \rho_{23} \\
    \rho_{13} & \rho_{23} &     1 
\end{bmatrix} = \begin{bmatrix}
   	\Sigma_{12}    & B \\
    B^{T} &  1 
\end{bmatrix} , ~~\text{where }~
\Sigma_{12} = \begin{bmatrix}
   	    1     & \rho_{12} \\
    \rho_{12} &     1 
\end{bmatrix} ~\text{and }~
B = \begin{bmatrix}
   	\rho_{13}\\
    \rho_{23} 
\end{bmatrix}. 
\]

We note that since $\Sigma$ and $\Sigma_{12}$ are covariance matrices, they must be positive definite, with $\det \left(\Sigma\right) = 1 - \rho_{12}^{2}-\rho_{13}^{2}-\rho_{23}^{2} + 2\rho_{12}\rho_{13}\rho_{23} > 0$ and $\det\left(\Sigma_{12}\right)=1-\rho_{12}^{2}>0$. The inverse of $\Sigma$ is given by the block matrix
\[
	\Sigma^{-1} = \begin{bmatrix}
   	\Sigma_{12}^{-1} + \Sigma_{12}^{-1}B(1-B^{T}\Sigma_{12}^{-1}B)^{-1}B^{T}\Sigma_{12}^{-1}~~    & -\Sigma_{12}^{-1}B(1-B^{T}\Sigma_{12}^{-1}B)^{-1} \\
    -(1-B^{T}\Sigma_{12}^{-1}B)^{-1}B^{T}\Sigma_{12}^{-1} & (1-B^{T}\Sigma_{12}^{-1}B)^{-1} 
\end{bmatrix}, 
\]
so that
\begin{align*}
\bmrem{1}_{3}^{T}\Sigma^{-1}\bmrem{1}_{3} &= \bmrem{1}_{2}^{T}\Sigma_{12}^{-1}\bmrem{1}_{2} +  (1-B^{T}\Sigma_{12}^{-1}B)^{-1} \left(1-\bmrem{1}_{2}^{T}\Sigma_{12}^{-1}B-B^{T}\Sigma_{12}^{-1}\bmrem{1}_{2}+\bmrem{1}_{2}^T\Sigma_{12}^{-1}BB^{T}\Sigma_{12}^{-1}\bmrem{1}_{2}\right)\\
&\hspace{-1cm}=\bmrem{1}_{2}^{T}\Sigma_{12}^{-1}\bmrem{1}_{2} + \frac{1-\rho_{12}^{2}}{1 - \rho_{12}^{2}-\rho_{13}^{2}-\rho_{23}^{2} + 2\rho_{12}\rho_{13}\rho_{23} } \left(1-2\bmrem{1}_{2}^T\Sigma_{12}^{-1}B+  \bmrem{1}_{2}^T\Sigma_{12}^{-1}BB^{T}\Sigma_{12}^{-1}\bmrem{1}_{2}\right)\\
&\hspace{-1cm}=\bmrem{1}_{2}^{T}\Sigma_{12}^{-1}\bmrem{1}_{2} + \frac{\det(\Sigma_{12})}{\det(\Sigma)}\left(1-\bmrem{1}_{2}^{T}\Sigma_{12}^{-1}B\right)^{2}\\
&\hspace{-1cm}=\bmrem{1}_{2}^{T}\Sigma_{12}^{-1}\bmrem{1}_{2} + \frac{\det(\Sigma_{12})}{\det(\Sigma)} \left(1-\frac{\rho_{13}+\rho_{23}}{1+\rho_{12}}\right)^{2} \geq \bmrem{1}_{2}^{T}\Sigma_{12}^{-1}\bmrem{1}_{2} ,
\end{align*}
with equality if and only if $1+\rho_{12}=\rho_{13}+\rho_{23}$. By similar calculations,
\[
	\bmrem{1}_{3}^{T}\Sigma^{-1}\bmrem{1}_{3} \geq \bmrem{1}_{2}^T\Sigma_{i,j}^{-1}\bmrem{1}_{2},~~~i<j \in\{1,2,3\},
\]
with equality if and only if $1+\rho_{ij}=\rho_{ik}+\rho_{jk}$, in which case $\eta_{1,2,3}=\eta_{i,j}$. Applying Theorem~\ref{thm:etatau2}, for this trivariate case, if $1+\rho_{C}\neq \sum\limits_{C':|C'|=2, C'\neq C}\rho_{C'}$ for all $C\subset\{1,2,3\}$ with $|C|=2$, then $\tau_{C}(1)=\eta_{C}$ for any set $C\subseteq\{1,2,3\}$ with $|C|\geq 2$. Since $\tau_C(\delta)$ is non-decreasing in $\delta$, $\tau_C(\delta)\leq\eta_C$ for $\delta\in[0,1)$. We also know $\tau_{1,2,3}=\eta_{1,2,3}$.

In this case, calculation of the explicit formulas for $\delta<1$, in a manner similar to the inverted logistic case, is complicated by the need to consider asymptotic approximations of Gaussian cumulative distribution functions beyond first order. As such, we do not attempt this here, but note that since $\tau_C(1)\leq\eta_C$, $|C|\geq 2$, we would be likely to estimate $\tau_C(\delta)<1$ in practice.
 
To gain some insight into the remaining cases of $C=\{i\}$, $i=1,2,3$, we consider the conditional extreme value model of \cite{Heffernan2004}. Let $\bmrem{Y}=\log \bmrem{X}$, so that $\bmrem{Y}$ has standard Gumbel marginal distributions, and all correlations be positive. Then conditioning on $Y_i$ gives
\[
	\text{pr}\left(Y_j - \rho^2_{ij}t \leq t^{1/2}z_j, Y_k - \rho^2_{ik}t \leq t^{1/2}z_k, Y_i>t\right)\sim N(z_j,z_k)\text{pr}(Y_i>t),~~~t\rightarrow\infty,
\]
for $N(z_j,z_k)$ denoting the distribution function of a particular Gaussian distribution (Heffernan and Tawn, 2004, Section~8). For $z_j=z_k=0$, this equates to
\begin{align}
	\text{pr}\left(Y_j \leq \rho^2_{ij}t, Y_k \leq \rho^2_{ik}t, Y_i>t\right)\sim N(0,0)\text{pr}(Y_i>t),~~~t\rightarrow\infty.
\label{eqn:MVNeqn}
\end{align}
In the trivariate case with Gumbel margins, equation~\eqref{eqn:RVassumption1} of Assumption~\ref{assumption:HRV} can be written as
\[
	\text{pr}\left(Y_j \leq  \delta t, Y_k \leq \delta t, Y_i>t\right)\in{\rm{RV}}_{-1/\tau_C(\delta)},~~~t\rightarrow\infty.
\]
Considering equation~\eqref{eqn:MVNeqn} again, we see that if $\delta\geq\max\left(\rho^2_{ij},\rho^2_{ik}\right)$, in place of the limiting Gaussian distribution, mass will occur at $(-\infty,-\infty)$ for variables $(Y_j,Y_k)$, which implies that $\tau_i(\delta)=1$. Alternatively, if $\rho^2_{ij}\leq\delta<\rho^2_{ik}$ or $\rho^2_{ik}\leq\delta<\rho^2_{ij}$, we have mass at $(-\infty,\infty)$ or $(\infty,-\infty)$, respectively, and for $\delta<\min\left(\rho^2_{ij},\rho^2_{ik}\right)$ mass occurs at $(\infty,\infty)$. In these cases, the left-hand side of equation~\eqref{eqn:MVNeqn} is $o\left\{\text{pr}\left(Y_i>t\right)\right\}$, which is consistent with $\tau_i(\delta)<1$.


\section{Simulation study}
\subsection{Estimation of $\tau_C(\delta)$ in Method~2}
In Method~2, introduced in Section~\ref{subsec:method2} of the paper, we consider regions of the form
\[
\tilde{E}_C = \left\{\bmrem{x}\in\mathbb{E}:x_\vee^{D\setminus C} \leq \left(x_\wedge^C\right)^\delta \right\}, ~|C|<d;~~~\tilde{E}_D = \mathbb{E}\setminus \bigcup_{C\in 2^D\setminus\emptyset:|C|<d}\tilde{E}_C,
\]
for $C\in2^D\setminus\emptyset$, and assume that $\text{pr}\left(X_\wedge^C >q, \bmrem{X}\in\tilde{E}_C\right)\in{\rm{RV}}_{-1/\tilde\tau_C(\delta)}$, for $X_\wedge^C = \min_{i\in C} X_i$. This is used as an alternative to considering $\text{pr}\left(X_\wedge^C>t,X_\vee^{D\setminus C}\leq t^\delta\right)\in {\rm{RV}}_{-1/\tau_C(\delta)}$, for which there is no clear structure variable. In Fig.~\ref{fig:tauEstimates}, we demonstrate how well the parameter $\tau_C(\delta)$ is approximated using Method~2. We consider the trivariate logistic distribution, with theoretical $\tau_C(\delta)$ values given in case~(iii) of Table~\ref{table:tauExamples}. For $\alpha=0.25$ and $\alpha=0.5$, we take samples of size 100,000 from this distribution, and use Method~2 to estimate $\tilde\tau_1(\delta)$, $\tilde\tau_{1,2}(\delta)$ and $\tilde\tau_{1,2,3}$ for values of $\delta\in\{0.1,\dots,0.95\}$. The thresholds used correspond to the 0.985 quantile of observed $X_\wedge^C$ values in each region $\tilde{E}_C$. Each simulation is repeated 100 times, and the true $\tau_C(\delta)$ parameter values are shown in grey.

The results indicate that the estimator derived from considering $\text{pr}\left\{X_\wedge^C>t,X_\vee^{D\setminus C}\leq \left(X_\wedge^C\right)^\delta\right\}$ yields $\tau_C(\delta)$ as defined through $\text{pr}\left(X_\wedge^C>t,X_\vee^{D\setminus C}\leq t^\delta\right)$. We observe increased variability in the estimates of $\tilde\tau_{1}(\delta)$ and $\tilde\tau_{1,2}(\delta)$ for small values of $\delta$, and $\tilde\tau_{1,2,3}$ for large values of $\delta$, which is likely due to the limited data in the corresponding regions $\tilde{E}_C$ for these cases. 

\begin{figure}[!htbp]
\begin{center}
\includegraphics[width=\textwidth]{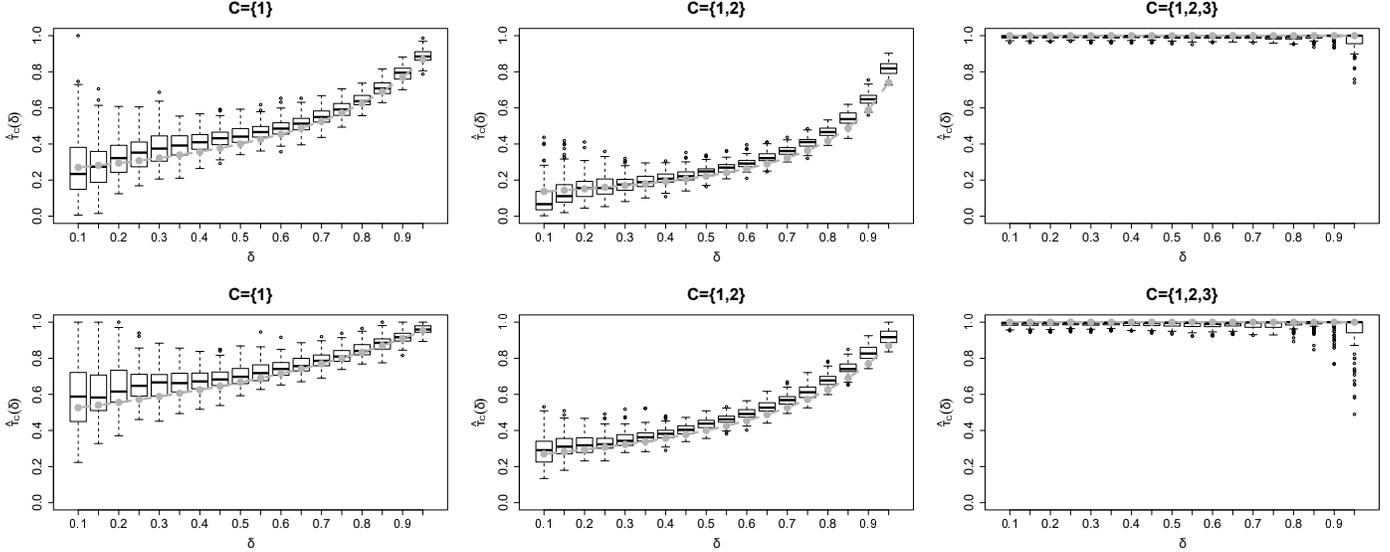}
\caption{Estimates of $\tau_1(\delta)$, $\tau_{1,2}(\delta)$ and $\tau_{1,2,3}$ for data simulated from trivariate logistic distributions with $\alpha=0.25$ (top) and $\alpha=0.5$ (bottom).}
\label{fig:tauEstimates}
\end{center}
\end{figure}

\subsection{Area under the receiver operating characteristic curve results for the max-mixture distribution}
In Table~\ref{table:AUCmix} of the main paper, we present the average area under the receiver operating characteristic curve for Method~1, Method~2 and the approach of \citeauthor{Goix2015}\ applied to samples taken from a particular max-mixture distribution. Figure~\ref{fig:AUC_mix} provides boxplots of all the values obtained in these simulations. 

\begin{figure}[!htbp]
\begin{center}
\includegraphics[width=\textwidth]{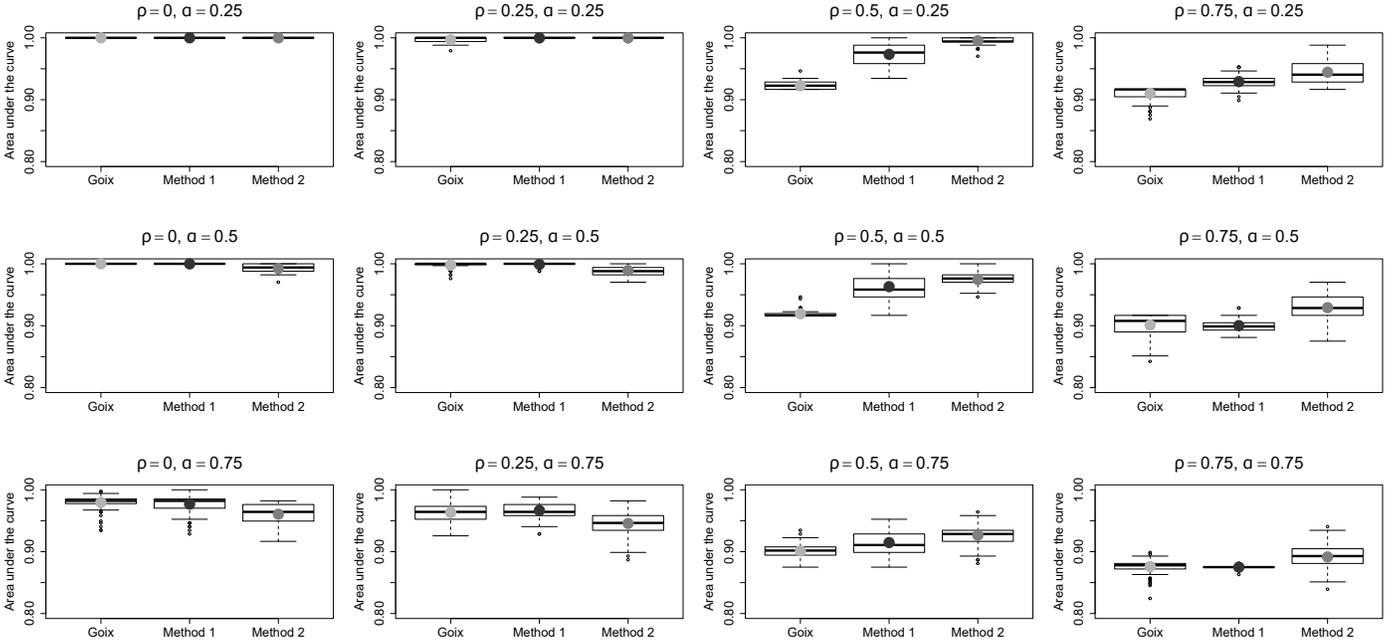}
\caption{Area under the receiver operating characteristic curve results for 100 simulations from a five-dimensional max-mixture distribution.}
\label{fig:AUC_mix}
\end{center}
\end{figure}

\subsection{Asymmetric logistic distribution}\label{subsec:simulationAsymmetric}
We now present simulation results for the asymmetric logistic distribution, using the same metrics as for the max-mixture distribution in Section~\ref{sec:simulation} of the paper. This model belongs to the class of multivariate extreme value distributions, and  it is possible to calculate the proportion of extremal mass associated with each cone $\mathbb{E}_C$.

In standard Fr\'{e}chet margins, the multivariate extreme value distribution function is of the form $\exp\left\{-V(\bmrem{x})\right\}$. \cite{ColesAndTawn1991} show that the spectral density corresponding to cone $\mathbb{E}_C$ is $h_{C}\left(\bmrem{w}_{C}\right) = -V^{\{C\}}(\bmrem{w}_{C})/d$, where
\begin{align*}
	V^{\{C\}}(\bmrem{x}_C)=\lim\limits_{x_{j}\rightarrow 0:j\in D\setminus C }\left(\prod_{i\in C}\frac{\partial}{\partial x_{i}}\right)V(\bmrem{x}),~~~~~\text{ $\bmrem{x}_{C}=\{x_{i}:i\in C\}$,}
\end{align*}
for $\bmrem{w}_{C}=\bmrem{x}_{C}/r_{C}$ and $r_{C}=\sum_{i\in C}x_{i}$. Hence, for $	\mathbb{B}_{C} = \{\bmrem{w}\in\mathcal{S}_{d-1} : w_{i} \in (0,1], i\in C;~w_{j}=0, j\in D\setminus C \}$, the proportion of mass on corresponding cone $\mathbb{E}_C$ is 
\begin{align}
p_{C}=-\frac{1}{d}\int_{\mathbb{B}_{C}}V^{\{C\}}(\bmrem{w}_{C})\prod_{i\in C}dw_{i}.
\label{eqn:spectralDensity}
\end{align}
For the asymmetric logistic model (\citeauthor{Tawn1990}, 1990), the exponent measure $V$ is defined as
\begin{align}
	V(\bmrem{x})= \sum_{C\in 2^{D}\setminus\emptyset} \left\{ \sum_{i\in C} \left(\theta_{i,C}/x_{i}\right)^{1/\alpha_{C}} \right\}^{\alpha_{C}},~~~~~\text{$\theta_{i,C} \in [0,1]$,}
\label{eqn:alogExponent}
\end{align}
with $\theta_{i,C}=0$ if $i\notin C$, $\sum\limits_{C\in 2^{D}\setminus\emptyset}\theta_{i,C}=1$ for all $i=1,\dots,d$ and $C\in 2^{D}\setminus\emptyset$, and dependence parameters $\alpha_{C}\in(0,1]$. In Proposition~1 of Section~\ref{subsec:alogProof}, we show that for the $d$-dimensional asymmetric logistic model with all $\alpha_{C}\equiv\alpha$, the proportion of mass on cone $\mathbb{E}_C$ is 
\[
	p^{(d)}_{C} = \sum_{i\in C}\theta_{i,C}/d,~~~~~\text{$C\in 2^{D}\setminus\emptyset$.}
\]
Using this new result, we can compare our estimated proportions to the truth using the Hellinger distance defined in equation~\eqref{eqn:hellinger} of the paper.

Following \cite{Goix2015}, we simulate data from an asymmetric logistic distribution with $\alpha_C\equiv\alpha$, whose extremal mass is concentrated on $f$ randomly chosen sub-cones, ensuring that moment constraint \eqref{eqn:momentConstraint} is satisfied. Suppose the sub-cones chosen correspond to subsets $F_{1},\dots,F_{f} \in 2^{D}\setminus\emptyset$. The conditions on the parameters of the asymmetric logistic distribution are satisfied by setting 
\begin{align}
	\theta_{i,C} = \left| \left\{ j : i\in F_{j},  j\in\{1,\dots,f\}\right\} \right|^{-1},~~~~~\text{$C\in \{F_{1},\dots,F_{f}\}$,}
\label{eqn:alogTheta}
\end{align}
and $\theta_{i,C}=0$ otherwise. We present results for dimensions $d=5$ and $d=10$. For $d=5$, we simulate samples of size $n=10,000$, and test both our methods when there are truly 5, 10 and 15 cones $\mathbb{E}_C$ with extremal mass. For $d=10$, we have $n=100,000$ samples, and consider 10, 50 and 100 cones with extremal mass. We set the tuning parameters as in Section~\ref{subsec:simulationAsymmetric} of the paper, and repeat each setting 100 times. In Table \ref{table:AUC}, we present the average area under the receiver operating characteristic curve for $\alpha\in\{0.25,0.5,0.75\}$. Boxplots of the full results obtained are provided in Figs.~\ref{fig:AUC_5D}~and~\ref{fig:AUC_10D}. In the asymmetric logistic model, the closer $\alpha_{C}$ is to 1, the larger the values of $\tau_{\underline{C}}(\delta)$ for any fixed $\delta$ and $\underline{C}\subset C$, as demonstrated by cases (ii) and (iii) in Table \ref{table:tauExamples}. Thus, the larger the value of $\alpha$ in our simulations, the more difficult it is to determine which  cones $\mathbb{E}_C$ truly contain extremal mass.

\begin{figure}
\begin{center}
\includegraphics[width=0.9\textwidth]{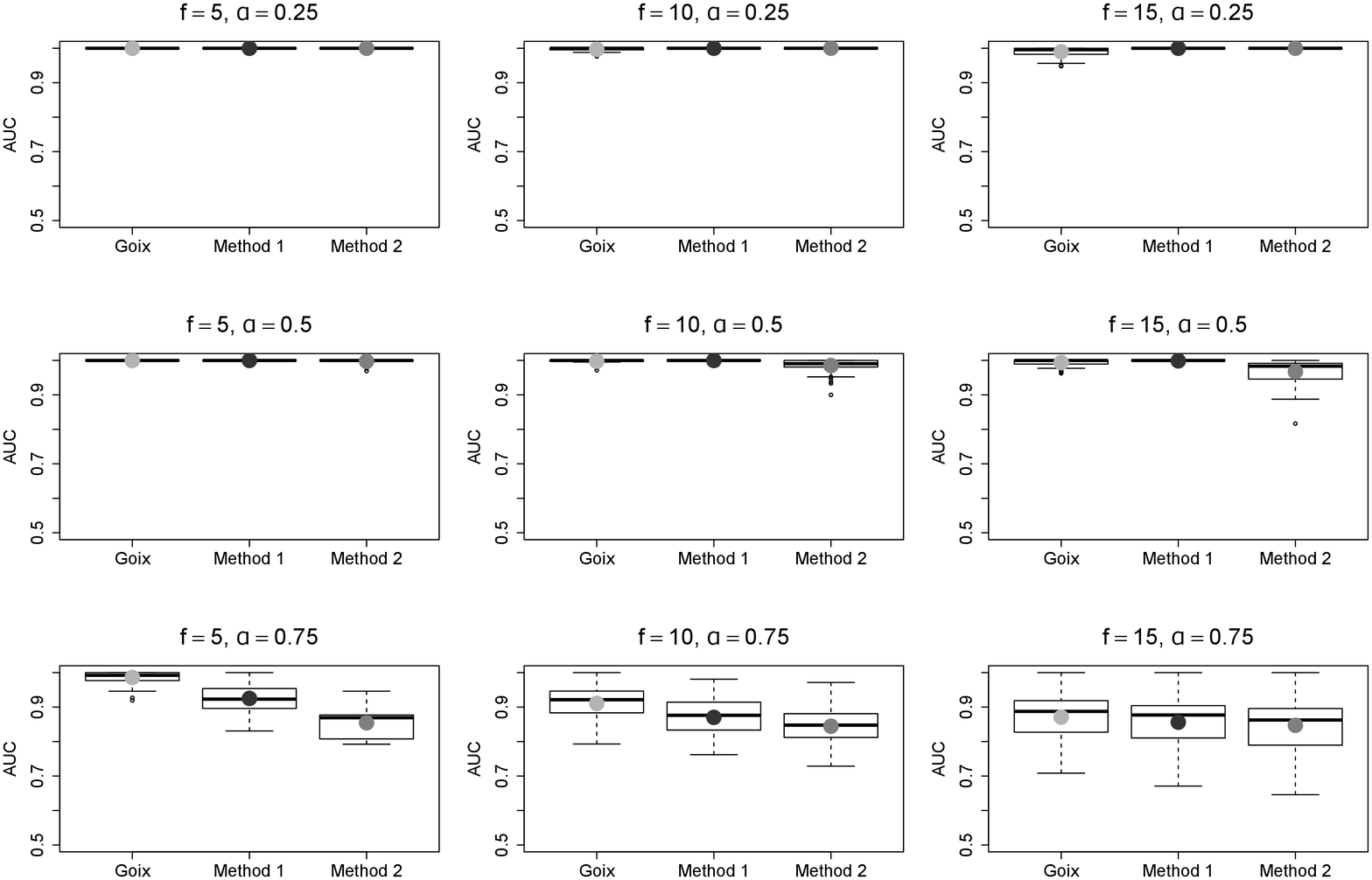}
\caption{Areas under the receiver operating characteristic curves for 100 simulations from a five-dimensional asymmetric logistic distribution.}
\label{fig:AUC_5D}
\includegraphics[width=0.9\textwidth]{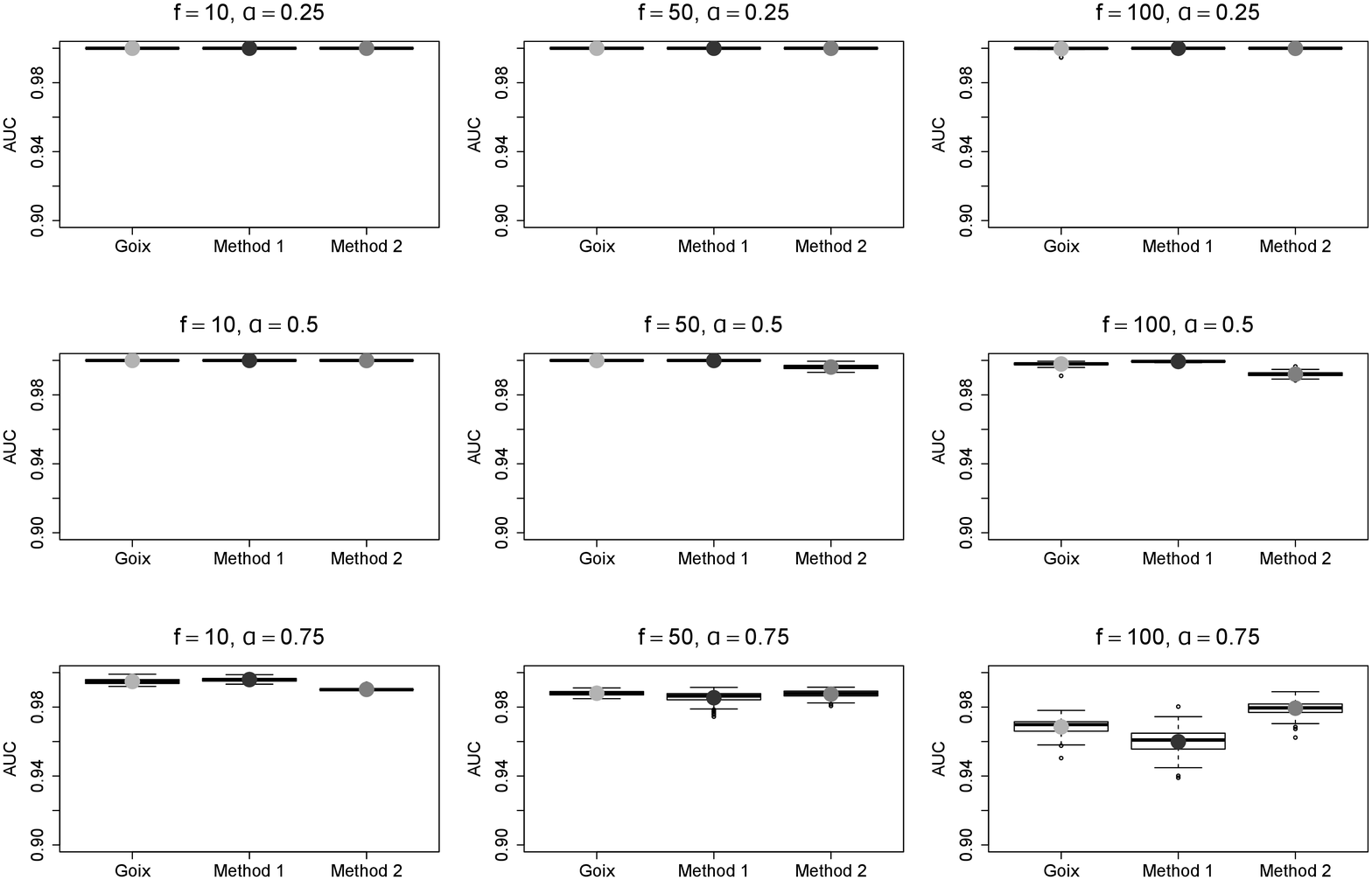}
\caption{Areas under the receiver operating characteristic curves for 100 simulations from a ten-dimensional asymmetric logistic distribution.}
\label{fig:AUC_10D}
\end{center}
\end{figure}

\begin{table}[!htbp]
\resizebox{\textwidth}{!}{%
\centering
\begin{tabular}{|c|ccccccccc|}
\hline
$(\alpha,f)$& $(0.25,5)$ & $(0.25,10)$ & $(0.25,15)$ & $(0.5,5)$ & $(0.5,10)$ & $(0.5,15)$ & $(0.75,5)$ & $(0.75,10)$ & $(0.75,15)$ \\
\hline
\citeauthor{Goix2016} & 100 (0.1) & 99.7 (0.5) & 99.0 (1.3) & 100 (0.2) & 99.8 (0.5) & 99.4 (0.9) & 98.6 (1.7) & 91.1 (4.9) & 87.1 (7.1)\\
Method~1		      & 100 (0.0) & 100 (0.0) & 100 (0.0) & 100 (0.0) & 100 (0.1) & 99.9 (0.3) & 92.5 (4.1) & 87.0 (5.5) & 85.6 (7.3)\\
Method~2			  & 100 (0.0) & 100 (0.0) & 100 (0.0) & 99.7 (0.6) & 98.5 (1.8) & 96.8 (3.2) & 85.4 (4.3) & 84.4 (5.7) & 84.7 (7.5)\\
\hline 
\hline
$(\alpha,f)$& $(0.25,10)$ & $(0.25,50)$ & $(0.25,100)$ & $(0.5,10)$ & $(0.5,50)$ & $(0.5,100)$ & $(0.75,10)$ & $(0.75,50)$ & $(0.75,100)$ \\
\hline
\citeauthor{Goix2016} &  100 (0.0) & 100 (0.0) & 100 (0.1) & 100 (0.0) & 100 (0.0) & 99.8 (0.1) & 99.5 (0.2) & 98.8 (0.1) & 96.9 (0.5)\\
Method~1		  	  &  100 (0.0) & 100 (0.0) & 100 (0.0) & 100 (0.0) & 100 (0.0) & 99.9 (0.0) & 99.6 (0.1) & 98.5 (0.4) & 96.0 (0.8)\\
Method~2			  &  100 (0.0) & 100 (0.0) & 100 (0.0) & 100 (0.0) & 99.6 (0.1) & 99.2 (0.1) & 99.0 (0.1) & 98.8 (0.2) & 97.9 (0.4)\\
\hline
\end{tabular}}
 \caption{Average areas under the receiver operating characteristic curves, given as percentages, for 100 samples from five-dimensional (top) and ten-dimensional (bottom) asymmetric logistic distributions, with dependence parameter $\alpha$ and $\theta_{i,C}$ determined via \eqref{eqn:alogTheta}. Standard deviations of these results are given in brackets.}
   \label{table:AUC}
\end{table}

\begin{figure}
\begin{center}
\includegraphics[width=\textwidth]{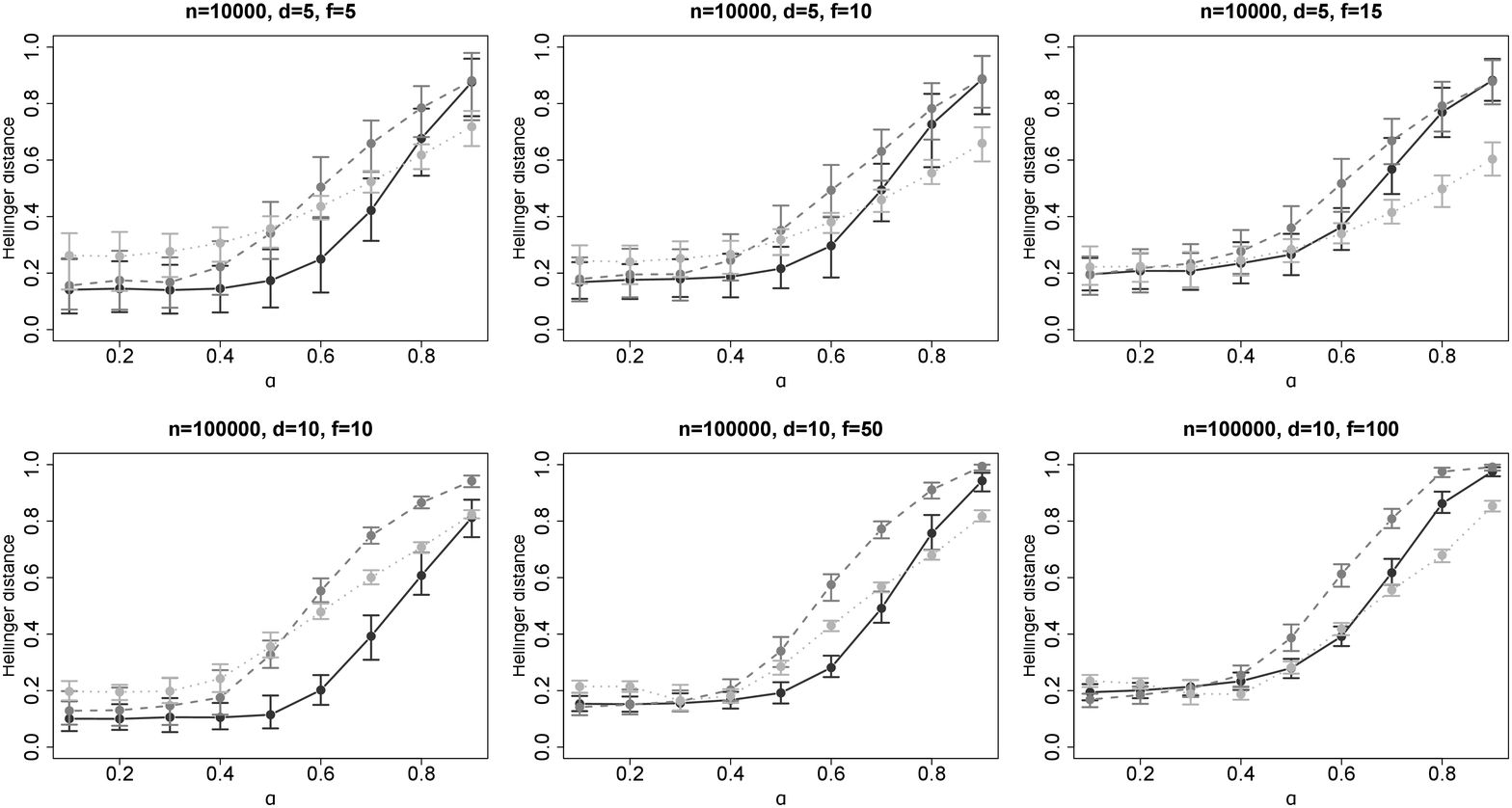}
\caption{Mean Hellinger distance, 0.05 and 0.95 quantiles over 100 simulations. Method~1: solid lines; Method~2: dashed lines; \citeauthor{Goix2015}: dotted lines.}
\label{fig:asymmetricSimulations}
\end{center}
\end{figure}

The average areas under the receiver operating characteristic curves in Table~\ref{table:AUC} show that all three methods perform well when $\alpha=0.25$ and $\alpha=0.5$, for both $d=5$ and $d=10$, with Method~1 slightly outperforming the other two approaches. The results suggest that the method of \cite{Goix2015} is generally the most successful classifier when $\alpha=0.75$, followed by Method~1, although the most substantial difference in results occurs for $(d,f,\alpha)=(10,100,0.75)$, where Method~2 is most successful; this is supported by the boxplots of results in Figs.~\ref{fig:AUC_5D}~and~\ref{fig:AUC_10D}. In principle, Method~1 should be better than Method~2 here, so greater assessment of tuning parameters may be required. It is possible that the method of \citeauthor{Goix2015}\ is most successful for larger values of $\alpha$ since it is more difficult for Methods~1~and~2 to distinguish between regions where $\tau_C(\delta)$ does and does not equal 1 in these cases. 

Figure~\ref{fig:asymmetricSimulations} shows the average Hellinger distance for $\alpha\in[0.1,0.9]$ for each of the cases described above. For the most sparse cases, $(d,f)=(5,5)$ and $(d,f)=(10,10)$, Method~1 performs significantly better than the other two approaches overall. For the less sparse cases, $(d,f)=(5,15)$ and $(d,f)=(10,100)$, the three methods give similar results in terms of the Hellinger distance for $\alpha\leq 0.5$, but the method of \citeauthor{Goix2015} is most successful for larger $\alpha$ values. When the extreme values are concentrated on fewer cones $\mathbb{E}_C$, it may be easier to estimate true values of $\tau_C(\delta)=1$ using Method~1 than in less sparse examples. For the asymmetric logistic distribution, Method~1 often performs better in terms of estimating the proportion of extremal mass on each cone $\mathbb{E}_C$, while the method of \cite{Goix2015} is better at classification, although this method does tend to place mass on too many cones $\mathbb{E}_C$, as shown in Fig.~\ref{fig:maxmixFaces} of the paper. 


\section{Calculating the mass on each cone $\mathbb{E}_C$ for an asymmetric logistic model}\label{subsec:alogProof}
\begin{proposition}
For the $d$-dimensional asymmetric logistic model with exponent measure \eqref{eqn:alogExponent} and $\alpha_{C}\equiv\alpha\in(0,1)$ for all $C\in 2^{D}\setminus\emptyset$, 
\[
d p_{C}^{(d)}=\sum\limits_{i\in C}\theta_{i,C},
\]
where $p_C^{(d)}$ denotes the proportion of mass on cone $\mathbb{E}_C$.
\end{proposition}

\begin{proof}
Consider the exponent measure of the asymmetric logistic model $V(\bmrem{x})$ as a sum of functions $V_{C}(\bmrem{x}_{C})$, for $C\in 2^{D}\setminus\emptyset$, i.e.,\
\[
V(\bmrem{x}) = \sum_{C\in 2^{D}\setminus\emptyset} V_{C}(\bmrem{x}_{C}), ~~~~V_{C}(\bmrem{x}_{C}) = \Bigg\{ \sum_{i\in C} \bigg(\frac{\theta_{i,C}}{x_{i}}\bigg)^{1/\alpha} \Bigg\}^{\alpha}.
\]
Then for any dimension $d\geq |C|$,
\begin{align*}
V^{\{C\}}(\bmrem{x}_{C}) = \left(\prod_{i\in C}\frac{\partial}{\partial x_{i}}\right) V_{C}(\bmrem{x}_C)= \left\{ \prod_{i=0}^{|C|-1}-\bigg(\frac{\alpha-i}{\alpha}\bigg) \right\} \left( \prod_{i\in C}\frac{\theta_{i,C}^{1/\alpha}}{x_{C,i} ^{1+1/\alpha}}\right)\left\{ \sum_{i\in C} \bigg(\frac{\theta_{i,C}}{x_{C,i}}\bigg)^{1/\alpha} \right\}^{\alpha-|C|},
\end{align*}
since for $\overline{C}\supset C$, $\lim\limits_{x_i\rightarrow 0:i\in \overline{C}\setminus C} \bigg(\prod\limits_{j\in C}\frac{\partial}{\partial x_{j}}\bigg)V_{\overline{C}}(\bmrem{x}_{\overline{C}})=0$. Hence, by result \eqref{eqn:spectralDensity},
\begin{align}
dp_{C}^{(d)} = -\int_{\mathbb{B}_{C}}V^{\{C\}}(\bmrem{w}_{C})\prod\limits_{i\in C}dw_{i},
\label{eqn:alogRes}
\end{align}
which we note does not depend on $d$. We claim that 
\begin{align}
-\int_{\mathbb{B}_{C}}V^{\{C\}}(\bmrem{w}_{C})\prod\limits_{i\in C}dw_{i} = \sum\limits_{i\in C}\theta_{i,C}.
\label{eqn:mainAlogRes}
\end{align}
First consider $|C|=1$, i.e.\ $\mathbb{B}_{C}=\{\bmrem{w}:w_{i}=1\}$ for $C=\{i\}$. Here,
\[
V^{\{i\}}(x_{i}) = \frac{\partial}{\partial x_{i}}V_{i}(x_{i}) = -\frac{\theta_{i,i}}{x_{i}^2},
\]
so 
\[
dp_{i}^{(d)} = \frac{\theta_{i,i}}{w_{i}^2}\Big|_{w_{i}=1} = \theta_{i,i},~~~~~~~ i=1,\dots,d.
\]
Now consider $|C|=2$. We have
\[
V^{\{i,j\}}(x_i,x_j) = \bigg(\frac{\alpha-1}{\alpha}\bigg) \frac{\big\{(1-\theta_{i,i})(1-\theta_{j,j})\big\}^{1/\alpha}}{(x_{i}x_{j})^{1+1/\alpha}}\bigg\{\bigg(\frac{\theta_{i,ij}}{x_{i}}\bigg)^{1/\alpha} + \bigg(\frac{\theta_{j,ij}}{x_{j}}\bigg)^{1/\alpha}\bigg\}^{\alpha-2},
\]
so 
\[
h_{i,j}(w_i)=\bigg(\frac{1-\alpha}{\alpha}\bigg) \frac{\big(\theta_{i,ij}\theta_{j,ij}\big)^{1/\alpha}}{\big\{w_{i}(1-w_{i})\big\}^{1+1/\alpha}}\bigg\{\bigg(\frac{\theta_{i,ij}}{w_{i}}\bigg)^{1/\alpha} + \bigg(\frac{\theta_{j,ij}}{1-w_{i}}\bigg)^{1/\alpha}\bigg\}^{\alpha-2},
\]
and
\begin{align}
dp_{i,j}^{(d)} = \int_{0}^{1}h_{i,j}(w_{i})dw_{i}.
\label{eqn:bivariateMass}
\end{align}
However, taking $d=2$, we know $2p_{1,2}^{(2)} + 2p_{1}^{(2)} + 2p_{2}^{(2)}=2$, so $dp_{1,2}^{(d)}=\theta_{1,12}+\theta_{2,12}$, and similarly, by \eqref{eqn:bivariateMass}, $\int_{0}^{1}h_{i,j}(w_{i})dw_{i}=dp_{i,j}^{(d)} =\theta_{i,ij}+\theta_{j,ij}$. So, \eqref{eqn:alogRes} holds for $|C|=1$ and $|C|=2$, and we suppose it holds for $|C|\leq k$, i.e.\ $dp_{C}^{(d)}=-\int_{\mathbb{B}_{C}}V^{\{C\}}(\bmrem{w}_{C})\prod\limits_{i\in C}dw_{i}=\sum\limits_{i\in C}\theta_{i,C}$.

Since \eqref{eqn:alogRes} does not depend on $d$, take $d=k+1$. So for all $\underline{C}$ with $|\underline{C}|\leq k$, $(k+1)p_{\underline{C}}^{(k+1)}=\sum\limits_{i\in\underline{C}}\theta_{i,\underline{C}}$. Now take $C=\{1,\dots,k+1\}$. Then,
\begin{align*}
(k+1)p_{C}^{(k+1)} &= (k+1) - \sum\limits_{\underline{C}\subset C}\sum\limits_{i\in \underline{C}}\theta_{i,\underline{C}} = (k+1) - \sum\limits_{i\in C}\sum\limits_{\underline{C}\subset C}\theta_{i,\underline{C}}\\
&=\sum\limits_{i\in C}\left(1-\sum\limits_{\underline{C}\subset C}\theta_{i,\underline{C}}\right)\\
&=\sum\limits_{i\in C}\theta_{i,C}.
\end{align*}
As such, \eqref{eqn:mainAlogRes} holds by induction.
\end{proof}

\bibliography{refs}
\end{document}